\pdfoutput=1
\documentclass[
    aps,
    pra,
    unsortedaddress,
    superscriptaddress,
    nofootinbib,
    floatfix,
    twocolumn,
    notitlepage
]{revtex4-2}

\usepackage{amsmath}
\usepackage{amssymb}
\usepackage{graphicx}
\usepackage{enumerate}
\usepackage{bbm}
\usepackage{xcolor}
\usepackage{latexsym}
\usepackage{mathtools}
\usepackage{comment}
\usepackage[vcentermath]{youngtab}
\usepackage[caption=false]{subfig}
\usepackage{paralist}
\usepackage{youngtab}
\usepackage{tabularx}
\usepackage{multirow}
\usepackage[export]{adjustbox}
\usepackage{braket}
\usepackage{xfrac}
\usepackage{centernot}

\usepackage{tikz}
\usepackage{pgfplots}
\pgfplotsset{compat=1.17}
\usepackage{tikz-3dplot}
\usetikzlibrary{arrows}
\usetikzlibrary{cd}
\usetikzlibrary{3d}

\usepackage[
    unicode=true,
    pdfusetitle,
    bookmarks=false,
    bookmarksnumbered=false,
    bookmarksopen=false,
]{hyperref}
\usepackage{amsthm}
\usepackage[capitalise]{cleveref} % https://stackoverflow.com/questions/6499504/cleveref-fails-for-theorem-environments-sharing-the-same-counter

\makeatletter
\def\th@plain{%
  \thm@notefont{}% same as heading font
  \itshape % body font
}
\def\th@definition{%
  \thm@notefont{}% same as heading font
  \normalfont % body font
}
\makeatother
\theoremstyle{plain}
\newtheorem{thm}{Theorem}%[section]
\newtheorem{lem}[thm]{Lemma}
\newtheorem{cor}[thm]{Corollary}

\newtheorem{prop}[thm]{Proposition}
\theoremstyle{definition}
\newtheorem{defn}{Definition}%[section]
\newtheorem{prob}{Problem}%[section]
\newtheorem{rem}{Remark}%[section]

\makeatletter
\if@cref@capitalise
\Crefname{thm}{Theorem}{Theorems}
\Crefname{lem}{Lemma}{Lemmas}
\Crefname{cor}{Corollary}{Corollaries}
\Crefname{conj}{Conjecture}{Conjectures}
\Crefname{defn}{Definition}{Definitions}
\Crefname{prob}{Problem}{Problems}
\Crefname{prop}{Proposition}{Propositions}
\Crefname{exam}{Example}{Examples}
\else
\crefname{thm}{theorem}{theorems}
\crefname{lem}{lemma}{lemma}
\crefname{cor}{corollary}{corollaries}
\crefname{conj}{conjecture}{conjectures}
\crefname{prob}{problem}{problems}
\crefname{defn}{definition}{definitions}
\crefname{prop}{proposition}{propositions}
\crefname{exam}{example}{examples}
\fi
\makeatother

\renewcommand{\epsilon}{\varepsilon} % getting bored of the standard epsilon
\newcommand{\dens}{\mathcal D}
\newcommand{\proj}{\mathcal P}

\newcommand{\sym}{\mathfrak{S}}
\newcommand{\symsub}{\vee}
\newcommand{\antisymsub}{\wedge}
\newcommand{\ident}{I}
\newcommand{\Tr}{\mathrm{Tr}}
\newcommand{\fd}{\delta} % finite difference
\newcommand{\abs}[1]{\left|#1\right|}
\newcommand{\norm}[1]{\left\lVert#1\right\rVert}
\newcommand{\s}{\mathcal}
\newcommand{\spec}{\mathrm{spec}}
\newcommand{\diag}{\mathrm{diag}}
\newcommand{\GL}{\mathrm{GL}}
\newcommand{\Unit}{\mathrm{U}}
\newcommand{\kl}[2]{D\!\left(#1\!\parallel\! #2\right)} % KL divergence
\newcommand{\qrl}[2]{S\!\left(#1\!\parallel\! #2\right)} % quantum relative entropy divergence
\newcommand{\keyl}[2]{K\!\left(#1\!\parallel\! #2\right)} % keyl divergence
\newcommand{\infkeyl}[1]{\Omega(#1)} % keyl divergence
\newcommand{\specht}[1]{W_{#1}} % notation for an Specht module for partition #1
\newcommand{\lilspecht}[1]{\omega_{#1}} % notation for an Specht module for partition #1
\newcommand{\schur}[1]{V_{#1}} % notation for a Schur module for partition #1
\newcommand{\lilschur}[1]{\pi_{#1}} % notation for a Schur module for partition #1
\newcommand{\yf}{\mathbb{Y}} % the set of Young frames
\newcommand{\spectra}{\mathbb{S}} % the set prob dists
\newcommand{\diff}{\mathop{}\!\mathrm{d}}

\newcommand{\lpm}{\mathrm{pm}} % leading principal minor

\newcommand{\strdia}[1]{\vcenter{\hbox{\includegraphics[scale=0.7]{figure-#1}}}}
\begin{document}

\title{A sufficient family of necessary inequalities for the compatibility of quantum marginals}

\author{Thomas C. Fraser}
\affiliation{Perimeter Institute for Theoretical Physics, 31 Caroline Street North, Waterloo, Ontario Canada N2L 2Y5}
\affiliation{Institute for Quantum Computing and Department of Physics and Astronomy, University of Waterloo, Waterloo, Ontario N2L 3G1, Canada}

\date{\today}

\begin{abstract}
    The quantum marginal problem is concerned with characterizing which collections of quantum states on different subsystems are compatible in the sense that they are the marginals of some multipartite quantum state.
    Presented here is a countable family of inequalities, each of which is necessarily satisfied by any compatible collection of quantum states.
    Additionally, this family of inequalities is shown to be sufficient: every incompatible collection of quantum states will violate at least one inequality belonging to the family.
\end{abstract}

\maketitle

\section{Introduction}
\label{sec:introduction}

The quantum marginal problem (QMP) is interested in characterizing the space of reduced/marginal states of a multipartite quantum state, and is widely regarded as being an important, albeit challenging problem to solve.

Connections and applications of the QMP to other topics in quantum theory (and beyond) include multipartite entanglement and separability~\cite{coffman2000distributed,walter2013entanglement}, quantum error correction~\cite{ocko2011quantum,huber2017quantum,yu2021complete}, entropic constraints~\cite{majenz2018constraints,christandl2018recoupling,kim2020entropy,osborne2008entropic}, other forms of quantum compatibility problems~\cite{heinosaari2016invitation,haapasalo2021quantum,doherty2008quantum}, the asymptotics of representations of the symmetric groups~\cite{daftuar2005quantum,christandl2006spectra,christandl2018recoupling}, the asymptotic restriction problem for tensors~\cite{christandl2018universal}, random matrix theory~\cite{collins2021projections,christandl2014eigenvalue}, and many-body physics~\cite{coleman2001reduced}.

The quantum marginal problem, and its associated terminology, is derived from an analogous problem in probability theory called the \textit{classical} marginal problem. The classical marginal problem is concerned with characterizing the relationships between the various marginal distributions of a joint, multivariate distribution~\cite{fritz2012entropic,vorob1962consistent,malvestuto1988existence}. As a joint probability distribution and its marginals can always be faithfully represented by the eigenvalues of a joint quantum state and its marginals using a product eigenbasis, the quantum marginal problem subsumes the classical marginal problem, and consequently, any of its applications~\cite{fritz2012entropic,liang2011specker,abramsky2011sheaf,fraser2018causal}.

One of the earliest formulations of the problem dates back to the early 1960s, when, for the purposes of simplifying calculations involving the atomic and molecular structure, quantum chemists became interested in characterizing the possible reduced density matrices of a system of $N$ interacting fermions~\cite{coulson1960present,coleman1963structure}. This version of the problem, referred to as the $N$-representability problem, has a long history~\cite{coleman2000reduced,coleman2001reduced,lude2013functional,borland1972conditions,ruskai2007connecting,klyachko2009pauli} that continues to evolve~\cite{mazziotti2012structure,mazziotti2012significant,klyachko2006quantum}.

Now the QMP comes in a variety of flavours which can be broadly organized by considering any additional assumptions or constraints that are imposed on either i) the properties of the joint state, and/or ii) the properties of the set of candidate density operators~\cite{tyc2015quantum}.

When focusing on the joint state, specializations of the QMP exist where the joint state is assumed to be fermionic~\cite{coleman2000reduced}, bosonic~\cite{wei2010interacting}, Gaussian~\cite{eisert2008gaussian,vlach2015quantum}, separable~\cite{navascues2021entanglement}, or having symmetric eigenvectors~\cite{aloy2020quantum}. Generally speaking, the QMP is difficult in the sense that it is a QMA-complete problem~\cite{liu2006consistency,liu2007quantum,wei2010interacting,bookatz2012qma}. For the purposes of this paper, the only restriction imposed on joint states will be that they live in a finite-dimensional Hilbert space, with a finite and fixed number of subsystems.

Regarding the list of candidate density operators, e.g., $(\rho_{AB}, \rho_{BC}, \rho_{AC})$, the list of subsystems they correspond to, e.g., $(AB, BC, AC)$, is known as the \textit{marginal scenario}, while the individual elements, e.g., $AB$, $BC$, or $AC$, will be referred to as \textit{marginal contexts}. A key consideration for understanding previous work on the QMP is whether the marginal contexts are disjoint. When the marginal contexts are disjoint, a complete solution to the QMP is known~\cite{klyachko2004quantum,klyachko2006general}. For a given specification of Hilbert space dimensions, this solution takes the form of a finite list of linear inequality constraints on the spectra of the candidate density operators. These solutions furthermore recover earlier results pertaining to low dimensional Hilbert-spaces~\cite{higuchi2003one,higuchi2003qutrit,bravyi2003requirements,han2005compatibility}.

In contrast, when the marginal scenario involves overlapping marginal contexts, existing results are comparatively more sporadic and typically weaker, being only applicable to low-dimensional systems, small numbers of parties, or only yielding necessary but insufficient constraints~\cite{chen2014symmetric,carlen2013extension,butterley2006compatibility,hall2007compatibility,chen2016detecting}. One promising approach, developed in~\cite{christandl2018recoupling} for relating marginal spectra to the recoupling theory of the symmetric group, appears limited to situations where the marginal contexts do not overlap too much.

Whenever a candidate set of density operators is explicitly given, one strategy to decide their compatibility is to use convex optimization techniques, e.g., semidefinite programming~\cite{vandenberghe1996semidefinite}. If the joint state is not necessarily pure, compatibility can be decided with a single semidefinite program~\cite{hall2007compatibility}. Additionally, when the joint state is assumed pure, compatibility can still be decided by an infinite hierarchy of semidefinite programs~\cite{yu2021complete}. In either case, analytic inequality constraints that serve as witnesses for incompatibility can be extracted from the outputs of such semidefinite programs~\cite{hall2007compatibility}.

The objective of this paper is to improve our understanding of the QMP, in particular for the case of overlapping marginal contexts, by i) deriving inequality constraints that are necessarily satisfied by all compatible collections of density operators, and ii) proving that if a collection of density operators satisfies these inequalities, then they are compatible. This paper begins by formally defining the QMP and then reformulating it from a different perspective. The primary advantage of this reformulation of the QMP is that it exposes an implicit symmetry of the problem which is helpful in deriving our main result.

\section{A Reformulation of the QMP}
\label{sec:reform_qmp}

This section introduces some notation and terminology that is sufficient to formally define both the QMP and an equivalent reformulation that is better suited for the techniques developed in subsequent sections.

First and foremost, every Hilbert space considered will be complex, finite-dimensional, and labeled by some subscript $X$, e.g., $H_{X}$. For each labeled Hilbert space, $H_{X}$, we will implicitly assume there exists some canonical orthonormal basis such that $H_{X} \cong \mathbb C^{d_X}$ where $d_X = \dim(H_{X})$. The corresponding set of linear operators, density operators (states), and pure states of $H_{X}$ are respectively denoted $\s L(H_{X})$, $\dens(H_{X})$, and $\proj(H_{X})$. The identity operator on $\s H_{X}$ is denoted $\ident_{X}$.

Given a list of labels, $S = (X_1, \ldots, X_k)$, which we identify with the concatenated string, $S \simeq X_1\cdots X_k$, the composite Hilbert space $H_{X_1} \otimes \cdots \otimes H_{X_k}$, will be labeled by $S$ itself and thus denoted $H_{S}$. For instance, if $S = ABC$, then $H_{ABC} = H_{A} \otimes H_{B} \otimes H_{C}$. Additionally, for any positive integer $n$ and label $X$, the list of labels consisting of $n$ copies of $X$ will be abbreviated by
\begin{equation}
    nX \coloneqq (X, \stackrel{n}{\ldots}, X) \simeq X\stackrel{n}{\cdots}X
\end{equation}
Using this notational convention, the $n$th tensor-power of a Hilbert space $H_{X}$ can be written as $H_{X}^{\otimes n} = H_{nX} = H_{X\stackrel{n}{\cdots}X}$.

Associated to any instance of the QMP, is a \textit{joint} (or \textit{global}) Hilbert space $H_{J}$ where $J = (X_1, \ldots, X_p)$ is a given finite list of labels called the \textit{joint context}. Every non-empty sublist $S$ of $J$ will be called a \textit{marginal context}. For each $S \subseteq J$, the partial trace from $H_{J}$ onto $H_{S}$ will be denoted by
\begin{equation}
    \label{eq:partial_trace}
    \Tr_{J \setminus S} : \s L(H_{J}) \to \s L(H_{S}).
\end{equation}
A finite, non-empty tuple of marginal contexts,
\begin{equation}
    \s M = (S_1, \ldots, S_m),
\end{equation}
is called a \textit{marginal scenario}. The cardinality of a marginal scenario will always be denoted by $m = \abs{\s M}$.
\begin{prob}[QMP]
    \label{prob:qmp}
    Given a marginal scenario, $\s M = (S_1, \ldots, S_m)$, and list of states $(\rho_{S_1}, \ldots, \rho_{S_m})$ where $\rho_{S_i} \in \dens(H_{S_i})$ for each $i \in \{1, \ldots, m\}$, decide if there exists a joint pure\footnotetext{The assumption of purity in the joint state can be made without loss of generality (see \cref{sec:pure_vs_mixed_qmp}).} state $\psi_{J} \in \proj(H_J)$ such that
    \begin{equation}
        \label{eq:qmp}
        \forall S_i \in \s M : \rho_{S_i} = \Tr_{J\setminus S_i}(\psi_{J}).
    \end{equation}
\end{prob}

Whenever such a pure state $\psi_{J}$ exists, the $m$-tuple of states $(\rho_{S_1}, \ldots, \rho_{S_m})$ is said to be \textit{compatible} and otherwise they are \textit{incompatible}.

For the sake of brevity, unless otherwise specified, a joint Hilbert space $H_{J}$ with joint context $J$ will always be implicitly given together with a particular marginal scenario $\s M = (S_1, \ldots, S_m)$ of length $m$.

Our first step is to reinterpret the $m$ linear constraints imposed on the joint state $\psi_{J}$ by \cref{eq:qmp} as a single linear constraint on the $m$th tensor-power state, $\psi_{J}^{\otimes m}$. Specifically, the $m$-tuple of states $(\rho_{S_1}, \ldots, \rho_{S_m})$ satisfies \cref{eq:qmp} if and only if
\begin{equation}
    \label{eq:parallel}
    \rho_{S_1} \otimes \cdots \otimes \rho_{S_m} = \Tr_{J \setminus S_1}(\psi_{J}) \otimes \cdots \otimes \Tr_{J \setminus S_m}(\psi_{J}).
\end{equation}
This observation motivates the following definitions.
\begin{defn}
    The \textit{Hilbert space on $\s M$}, denoted $H_{\s M}$, is
    \begin{equation}
        H_{\s M} \coloneqq H_{S_1} \otimes \cdots \otimes H_{S_m}.
    \end{equation}
    The \textit{partial trace from $mJ$ onto $\s M$}, denoted $\Tr_{mJ\setminus \s M}$, is
    \begin{equation}
        \Tr_{mJ\setminus \s M} \coloneqq \Tr_{J\setminus S_1} \otimes \cdots \otimes \Tr_{J\setminus S_m}.
    \end{equation}
\end{defn}
Note that the partial trace from $mJ$ onto $\s M$, $\Tr_{mJ\setminus \s M}$, is simply the partial trace operation mapping elements of $\s L(\s H_{mJ}) = \s L(H_{J}^{\otimes m})$ to elements of $\s L(H_{\s M}) = \s L(H_{S_1} \otimes \cdots \otimes H_{S_m})$.
\begin{defn}
    An \textit{$\s M$-product state} is any state, $\rho_{\s M} \in \dens(H_{\s M})$, of the form
    \begin{equation}
        \rho_{\s M} = \rho_{S_1} \otimes \cdots \otimes \rho_{S_m}
    \end{equation}
    where each component, $\rho_{S_i}$, is a state in $\dens(H_{S_i})$.
\end{defn}
Since there is a bijection between $m$-tuples of states $(\rho_{S_1}, \ldots, \rho_{S_m})$ and $\s M$-product states $\rho_{S_1} \otimes \cdots \otimes \rho_{S_m}$, the QMP can be equivalently restated entirely in terms of $\rho_{\s M}$.
\begin{prob}
    \label{prob:cqmp}
    Given an $\s M$-product state, $\rho_{\s M}$, determine whether or not there exists a pure state $\psi_{J}$ such that
    \begin{equation}
        \label{eq:cqmp}
        \rho_{\s M} = \Tr_{mJ\setminus \s M}(\psi_{J}^{\otimes m}).
    \end{equation}
\end{prob}
The equivalence between \cref{prob:qmp} and \cref{prob:cqmp} follows directly from \cref{eq:parallel}; moreover, whenever such a pure state $\psi_{J}$ exists in either formulation of the QMP, it satisfies both \cref{eq:qmp} and \cref{eq:cqmp}. Pursuant to this equivalence, an $\s M$-product state, $\rho_{S_1} \otimes \cdots \otimes \rho_{S_m}$, is said to be \textit{compatible} whenever the $m$-tuple of states $(\rho_{S_1}, \ldots, \rho_{S_m})$ is compatible (and \textit{incompatible} otherwise). The set of all compatible $\s M$-product states will be denoted $\s C_{\s M}$.

\section{Necessary and Sufficient Inequality Constraints}
\label{sec:necc_suff}

In this section, we construct inequalities that are necessarily satisfied by \textit{all} compatible $\s M$-product states, $\rho_{\s M}$.
These inequalities, therefore, can be used to answer the QMP in the negative; if an $\s M$-product state violates any of the forthcoming inequalities, then it is incompatible.
In addition, it will be shown that if an $\s M$-product state, $\rho_{\s M}$, satisfies all of the forthcoming inequalities, then it must be compatible.

These inequalities emerge from considering the permutation symmetry of the $k$th tensor power, $\psi_J^{\otimes k}$, of a pure state $\psi_J \in \proj(H_{J})$. For each $k \in \mathbb N$, let $\sym_{k}$ be the symmetric group on $k$ symbols, and let $T_{J} : \sym_{k} \to \s L(H_{J}^{\otimes k})$ be the representation of $\sym_k$ acting on $H_{J}^{\otimes k}$ by permutation of its $k$ factors.
\begin{defn}
    \label{defn:symsub}
    The $k$th symmetric subspace $\symsub^{k}H_{J} \subseteq H_{J}^{\otimes k}$ is defined as
    \begin{equation}
        \symsub^{k}H_{J} = \{ \ket{\phi} \in H_{J}^{\otimes k} \mid \forall \pi \in \sym_k, T_{J}(\pi) \ket{\phi} = \ket{\phi} \}.
    \end{equation}
    The orthogonal projection operator onto $\symsub^{k}H_{J}$ will be denoted by $\Pi_{J}^{(k)}$.
\end{defn}
Given any vector $\ket{\psi_{J}} \in H_{J}$, it is straightforward to verify that $\ket{\psi_{J}}^{\otimes k}$ is an element of the $k$th symmetric subspace $\symsub^{k}H_{J} \subseteq H_{J}^{\otimes k}$.
\begin{prop}
    \label{prop:copy_sym}
    Let $\psi_{J} = \ket{\psi_J}\bra{\psi_J}$ be a pure state, let $k \in \mathbb N$. Then\footnotetext{Throughout this paper, an inequality $A \geq B$ between Hermitian operators $A$ and $B$ always indicates that $A - B$ is positive semidefinite, i.e. $A - B \geq 0$. See \cite[Section V]{bhatia1997matrix}.}
    \begin{equation}
        \label{eq:kth_inclusion}
        \psi_{J}^{\otimes k} \leq \Pi_{J}^{(k)},
    \end{equation}
    where $\Pi_{J}^{(k)}$ is defined in \cref{defn:symsub}.
\end{prop}
By comparing \cref{eq:cqmp} with \cref{eq:kth_inclusion}, and recalling that partial traces are positive channels, it becomes possible to \textit{eliminate} $\psi_{J}$ from \cref{eq:cqmp}. For example, when $k=m$, the partial trace $\Tr_{mJ \setminus \s M}$ applied to \cref{eq:kth_inclusion} yields $\Tr_{mJ\setminus\s M}(\psi_{J}^{\otimes m}) \leq \Tr_{mJ\setminus\s M}(\Pi_{J}^{(m)})$ and thus we obtain the following corollary.
\begin{cor}
    If $\rho_{\s M}$ is a compatible $\s M$-product state, then
    \begin{equation}
        \label{eq:n1}
        \rho_{\s M} \leq \Tr_{mJ\setminus\s M}(\Pi_{J}^{(m)}).
    \end{equation}
\end{cor}
In \cref{sec:n1}, it is shown that the utility of this constraint is quite sensitive to the marginal scenario under consideration. For certain marginal scenarios, \cref{eq:n1} happens to be satisfied by \textit{all} $\s M$-product states, and thus is useless for the purposes of the QMP. Nevertheless, for other marginal scenarios, \cref{eq:n1} happens to be violated by some $\s M$-product states, and thus is a \textit{non-trivial} condition for the compatibility of an $\s M$-product state $\rho_{\s M}$.

Analogous reasoning can be used to construct stronger inequality constraints for the QMP. When $k$ is a multiple of $m$, $k=nm$, one can apply the $n$th tensor power of $\Tr_{mJ\setminus\s M}$ to both sides of \cref{eq:kth_inclusion}. While it is clear from the preceding discussion that this will yield inequality constraints necessarily satisfied by all $\s M$-product states, we will additionally show that their satisfaction for all $n \in \mathbb N$ is \textit{sufficient} to conclude that $\rho_{\s M}$ must be an $\s M$-product state.
\begin{thm}
    \label{thm:main}
    An $\s M$-product state, $\rho_{\s M}$, is compatible if and only if for all $n \in \mathbb N$,
    \begin{equation}
        \label{eq:nth_order}
        \rho_{\s M}^{\otimes n} \leq \Tr_{mJ\setminus\s M}^{\otimes n}(\Pi^{(nm)}_{J}).
    \end{equation}
\end{thm}
Note that $\Tr_{mJ\setminus\s M}^{\otimes n}$ is the partial trace operation taking elements of $\s L(H_{J}^{\otimes nm})$ to elements of $\s L(H_{\s M}^{\otimes n})$.

To prove \cref{thm:main} we use the following lemma, proven in \cref{sec:proof_finetti_lemma}, that the upper-bound in \cref{eq:nth_order}, up to normalization, represents the expected value of $\sigma_{\s M}^{\otimes n}$ when $\sigma_{\s M}$ is sampled according to a probability measure, $\nu_{\s M}$, whose support is precisely the set of compatible $\s M$-product states, denoted $\s C_{\s M}$.
\begin{lem}
    \label{lem:de_finetti_compatible}
    There exists a probability measure, $\nu_{\s M}$, over $\dens(H_{\s M})$, with support $\s C_{\s M}$, such that for all $n \in \mathbb N$,
    \begin{equation}
        \Tr_{mJ\setminus\s M}^{\otimes n}(\Pi^{(nm)}_J) = \tbinom{nm+d_{J}-1}{nm} \int_{\s C_{\s M}} \hspace{-1em}\nu_{\s M}(\diff \sigma_{\s M}) \sigma_{\s M}^{\otimes n},
    \end{equation}
    where $d_J = \dim(H_J)$.
\end{lem}
\begin{proof}
    See \cref{sec:proof_finetti_lemma}.
\end{proof}
Now consider a measurement effect, $E_{n}$, acting on $H_{\s M}^{\otimes n}$, i.e. a Hermitian operator $E_n \in \s L(H_{\s M}^{\otimes n})$ such that $0 \leq E_n \leq \ident_{\s M}^{\otimes n}$. If \cref{eq:nth_order} is satisfied by some state $\rho_{\s M} \in \dens(H_{\s M})$, and $\Tr(E_n\rho_{\s M}^{\otimes n}) \neq 0$, then \cref{lem:de_finetti_compatible} implies
\begin{equation}
    \label{eq:prob_ratio_lower}
    \sup_{\sigma_{\s M} \in \s C_{\s M}}\frac{\Tr(E_n\sigma_{\s M}^{\otimes n})}{\Tr(E_n\rho_{\s M}^{\otimes n})} \geq \tbinom{nm+d_{J}-1}{nm}^{-1}.
\end{equation}
Therefore, to show that every $\s M$-product state, $\rho_{\s M}$, eventually violates \cref{eq:nth_order}, and thus prove \cref{thm:main}, it suffices to prove the existence of a sequence of measurement effects, $n \mapsto E_n$, such that the above ratio of probabilities, as a function of increasing $n$, approaches zero faster than $\tbinom{nm+d_{J}-1}{nm}^{-1}$, thus violating \cref{eq:prob_ratio_lower}.

The particular problem of finding a sequence of measurements such that $\Tr(E_n\rho_{\s M}^{\otimes n})$ stays reasonably large while simultaneously minimizing $\Tr(E_n\sigma_{\s M}^{\otimes n})$ for all $\sigma_{\s M}$ distinct from $\rho_{\s M}$ is related to the problems of asymmetric quantum state discrimination and quantum hypothesis testing~\cite{hiai1991proper,ogawa2005strong,pereira2022analytical,hayashi2001asymptotics,hayashi2002two}. Broadly speaking, existing results in these fields are sufficiently strong to establish the claimed violation of \cref{eq:prob_ratio_lower} for each incompatible $\rho_{\s M} \not \in \s C_{\s M}$. Our specific approach relies on ideas developed by \citeauthor{keyl2006quantum}~\cite{keyl2006quantum} for the purposes of quantum state estimation\footnote{Our estimation-theoretic approach is largely inspired by applications of the \textit{spectral} estimation technique~\cite{keyl2005estimating} to the \textit{spectral} QMP~\cite{christandl2006spectra,christandl2007nonzero,christandl2018recoupling}.}.

In the interest of being constructive and non-asymptotic, in \cref{sec:keyl_exclusive} (specifically \cref{rem:state_discrim}) we show how to construct, for each $\rho_{\s M}$, an explicit sequence of projection operators, $n \mapsto E_n$, such that for all $n \in \mathbb N$,
\begin{equation}
    \label{eq:incompatibility_measure}
    \sup_{\sigma_{\s M} \in \s C_{\s M}}\frac{\Tr(E_{n} \sigma_{\s M}^{\otimes n})}{\Tr(E_{n} \rho_{\s M}^{\otimes n})} \leq \exp( - (n-d_{\s M}^2) \infkeyl{\rho_{\s M}} + c(\rho_{\s M}))
\end{equation}
where $c(\rho_{\s M}) \geq 0$ and $\infkeyl{\rho_{\s M}} \geq 0$ are quantities independent of $n$ (but dependent on $\rho_{\s M}$) and $d_{\s M} = \dim(H_{\s M})$. Additionally, it is shown that the exponential rate, $\infkeyl{\rho_{\s M}}$, vanishes if and only if $\rho_{\s M}$ is compatible and thus its positivity can serve as a witness of the incompatibility of $\rho_{\s M}$.

\begin{proof}[Proof of \cref{thm:main}]
    The discussion preceding \cref{thm:main} has already established the ``only if'' portion of \cref{thm:main}: applying $\Tr_{mJ\setminus\s M}^{\otimes n}$ to \cref{eq:kth_inclusion} when $k = nm$ yields
    \begin{align}
        \rho_{\s M}^{\otimes n}
        = \Tr_{mJ\setminus\s M}^{\otimes n}(\psi_J^{\otimes nm})
        \leq \Tr_{mJ\setminus\s M}^{\otimes n}(\Pi^{(nm)}_{J}).
        \label{eq:nth_disco}
    \end{align}
    Therefore, all that remains is to prove the ``if'' portion of \cref{thm:main}.
    Suppose $\rho_{\s M} \in \dens(H_{\s M})$ is a state that satisfies \cref{eq:nth_order} for some particular value of $n$. By combining \cref{eq:incompatibility_measure} with \cref{eq:prob_ratio_lower}, we conclude
    \begin{align}
        \label{eq:incompatibility_upper_bound}
        \infkeyl{\rho_{\s M}}
        &\leq \frac{\ln \tbinom{nm + d_{J} - 1}{nm} + c(\rho_{\s M})}{n-d_{\s M}^{2}}.
    \end{align}
    Since $\tbinom{nm + d_{J} - 1}{nm} \in O(n^{d_{J}-1})$ is polynomial of degree $d_{J}-1$ in $n$, the upper bound above approaches zero in the limit as $n \to \infty$. For finite $n$, the inequality in \cref{eq:incompatibility_upper_bound} merely implies that $\infkeyl{\rho_{\s M}}$ must be small. For the purposes of \cref{thm:main}, if a state $\rho_{\s M}$ satisfies \cref{eq:nth_order} for all $n$, \cref{eq:incompatibility_upper_bound} implies that $\infkeyl{\rho_{\s M}} = 0$ and thus $\rho_{\s M} \in \s C_{\s M}$ must be a compatible $\s M$-product state.
\end{proof}

\section{Conclusion}
\label{sec:conclusion}

This paper makes progress toward an analytic solution to the quantum marginal problem (QMP) by constructing an countably-infinite family of necessary operator inequalities whose satisfaction by a given tuple of density operators is sufficient to conclude their compatibility.
The primary advantage of this approach is its generality: for any finite Hilbert space dimension(s) and any number of subsystems with arbitrary overlap, the corresponding family of necessary inequalities is shown to be sufficient.
The results of this paper, therefore, constitute the first quantifier-free solution to the QMP for overlapping marginal contexts.
However, the characterization of compatible density operators produced by this approach is not finite, and thus inherently more challenging to compute.
Evidently, further insights will be required to produce a finite set of necessary and sufficient conditions for the overlapping QMP.

\begin{acknowledgments}
    I would like to thank J. Davis and R. Spekkens for numerous helpful discussions.
    The majority of this research was conducted while T.C.F. was supported by the Natural Sciences and Engineering Research Council of Canada (NSERC), grant 411301803.
\end{acknowledgments}

\bibliographystyle{apsrev4-1}
\bibliography{references}

\appendix

\section{Schur-Weyl Decompositions}
\label{sec:schur_weyl}

Given any representation of the symmetric group $\sym_n$ over a finite-dimensional complex space, such as the aforementioned tensor-permutation representation $T : \sym_n \to \s L((\mathbb C^d)^{\otimes n})$, Maschke's Theorem guarantees the representation is \textit{completely reducible} and therefore decomposes into irreducible subrepresentations~\cite[Theorem 1.5.3]{sagan2013symmetric}. Furthermore, the complete set of non-isomorphic irreducible representations of $\sym_n$ is isomorphic to the set of conjugacy classes of $\sym_n$~\cite[Proposition 1.10.1]{sagan2013symmetric} which itself is isomorphic to the set of partitions of $n$.
\begin{defn}
    A \textit{partition of $n$}, $\lambda = (\lambda_1, \ldots, \lambda_\ell)$, is a sequence of non-increasing ($\lambda_i \geq \lambda_{i+1}$) positive integers ($\lambda_i \in \mathbb N$) whose total sum is $n$ (${\sum}_i \lambda_i = n$). The length of $\lambda$ is denoted by $\ell = \ell(\lambda)$. The \textit{set of all partitions of $n$} will be denoted $\yf_n$.
\end{defn}
For each partition $\lambda \in \yf_n$, let $\lilspecht{\lambda} : \sym_n \to \specht{\lambda}$ denote the corresponding irreducible representation of $\sym_n$, otherwise known as the Specht module for $\lambda$~\cite[Section 2.3]{sagan2013symmetric}. Using this notation, the Maschke decomposition of the tensor-permutation representation $T : \sym_n \to \s L((\mathbb C^d)^{\otimes n})$ yields a decomposition of $(\mathbb C^d)^{\otimes n}$,
\begin{align}
    (\mathbb C^{d})^{\otimes n} \cong {\bigoplus}_{\lambda \in \yf_n} \specht{\lambda} \otimes \schur{\lambda}^d,
    \label{eq:maschkes_theorem}
\end{align}
where the $\schur{\lambda}^{d}$ denotes the \textit{multiplicity space}, whose dimension counts the number of isomorphic copies of $\specht{\lambda}$ in $(\mathbb C^{d})^{\otimes n}$. It is also worth noting that $\dim(\schur{\lambda}^d) > 0$ if and only if $\ell(\lambda) \leq d$~\cite{sagan2013symmetric} and therefore the above summands over $\lambda$ are implicitly restricted to the subset $\yf^d_n \subseteq \yf_n$ of partitions of $n$ with length at most $d$. Another result, referred to as Schur-Weyl duality \cite[Chapter 9]{procesi2007lie}, implies that the multiplicity space $\schur{\lambda}^d$ itself supports an irreducible representation of $\GL(d)$, denoted $\lilschur{\lambda} : \GL(d) \to \s L(\schur{\lambda}^d)$. Let $\ket{\phi_{\lambda}} \in \schur{\lambda}^{d}$ denote the unique highest weight vector of $\schur{\lambda}^{d}$ characterized by the property that
\begin{equation}
    \lilschur{\lambda}(\diag(x_1, \ldots, x_d)) \ket{\phi_{\lambda}} = \prod_{i=1}^{d} x_i^{\lambda_i} \ket{\phi_{\lambda}}
\end{equation}
for all $\diag(x_1, \ldots, x_d) \in \mathrm{GL}(d)$.
Furthermore, for each partition $\lambda \in \yf^d_n$, let
\begin{align}
    \label{eq:iota_iso}
    \iota_{\lambda} : \specht{\lambda} \otimes \schur{\lambda}^d \xhookrightarrow{} (\mathbb C^d)^{\otimes n}
\end{align}
be the $\sym_n \times \GL(d)$-intertwining isometry from the isotypic subspace $\specht{\lambda} \otimes \schur{\lambda}^d$ into $(\mathbb C^d)^{\otimes n}$ associated to $\lambda$. Furthermore, let
\begin{equation}
    \Pi_{d}^{\lambda} \coloneqq (\iota_{\lambda})(\iota_{\lambda})^{\dagger}
\end{equation}
denote the corresponding orthogonal projection operator acting on $(\mathbb C^d)^{\otimes n}$.
\begin{prop}
    \label{prop:decomp_sym_operator}
    Let $Q \in \s L((\mathbb C^d)^{\otimes n})$ be an $\sym_n$-invariant operator in the sense that
    \begin{equation}
        \forall g \in \sym_n : T(g) Q T^{\dagger}(g) = Q.
    \end{equation}
    Then $Q$ admits of the following decomposition,
    \begin{equation}
        \label{eq:Q_decomp}
        Q = \bigoplus_{\lambda \in \yf^d_n} \ident_{\specht{\lambda}} \otimes \tau_{\lambda}(Q),
    \end{equation}
    where the $\lambda$-component of $Q$, $\tau_{\lambda}(Q) \in \s L(\schur{\lambda}^{d})$, is defined as
    \begin{equation}
        \tau_\lambda(Q) = \frac{\Tr_{\specht{\lambda}}(\iota_{\lambda}^{\dagger} Q \iota_\lambda)}{\dim(\specht{\lambda})}.
    \end{equation}
\end{prop}
\begin{proof}
    This follows from an application of Schur's lemma~\cite[Theorem 4.29]{hall2015lie}.
\end{proof}

\begin{defn}
    \label{defn:twirled}
    Let $\lilschur{\lambda} : \GL(d) \to \schur{\lambda}^d$ be the irreducible representation of $\GL(d)$ with highest weight vector $\ket{\phi_\lambda} \in \schur{\lambda}^d$ where $\lambda \in \yf_n^d$. For each unitary operator $U \in \Unit(d) \subseteq \GL(d)$, let the \textit{twirled highest weight vector} be defined as $\ket{\phi^{U}_\lambda} \coloneqq \lilschur{\lambda}(U) \ket{\phi_\lambda} \in \schur{\lambda}^{d}$.
\end{defn}
In \cref{sec:keyl_exclusive}, specifically \cref{prop:hwv_lpm}, we shall see that the quantity $\bra{\phi_\lambda^{U}} \tau_\lambda(\rho^{\otimes n}) \ket{\phi_{\lambda}^{U}}$, which depends only on $\lambda \in \yf_n^d$ and $U^{\dagger} \rho U \in \dens(\mathbb C^d)$, admits of a formula that remains well-defined even when $\lambda$ is permitted to be a non-increasing sequence of \textit{non-negative real numbers}.

\section{Spectra \& Partitions}
\label{sec:spectra_partitions}

The purpose of this section is to develop a connection between i) partitions $\lambda \in \yf_n^d$ with length at most $d$, and ii) the possible eigenvalues of density operators $\rho \in \dens(\mathbb C^d)$.
\begin{defn}
    A subset $C \subseteq \mathbb R^d$ is called a \textit{convex cone} if it is closed under
    \begin{enumerate}[i)]
        \item addition: for any $x, y \in C$, $x+y \in C$, and
        \item multiplication: for any $x \in C$, and $a \geq 0$, $a x \in C$.
    \end{enumerate}
\end{defn}
Two convex cones that are relevant here will be the cone of non-negative real numbers,
\begin{equation}
    \mathbb R_{\geq 0}^{d} = \{ (x_1, \ldots, x_d) \in \mathbb R^d \mid \forall i : x_i \geq 0\},
\end{equation}
and the subset of non-increasing non-negative real numbers,
\begin{equation}
    \mathbb R_{\geq 0}^{d;\downarrow} = \{ (x_1, \ldots, x_d) \in \mathbb R^d \mid x_1 \geq \cdots \geq x_d \geq 0\}.
\end{equation}
While there is a natural surjective map from $\mathbb R_{\geq 0}^{d}$ to $\mathbb R_{\geq 0}^{d;\downarrow}$ which sorts the elements of $(x_1, \ldots, x_d)$ in a non-increasing order, there is also a \textit{bijective} linear map $\gamma : \mathbb R_{\geq 0}^{d} \to \mathbb R_{\geq 0}^{d;\downarrow}$ which takes partial sums. Specifically, $\gamma$ maps $y = (y_1, \ldots, y_{d}) \in \mathbb R_{\geq 0}^{d}$ to $\gamma(y) = (\gamma_1(y), \ldots, \gamma_d(y))$ where
\begin{equation}
    \label{eq:cumm}
    \gamma_i(y) = y_i + y_{i+1} + \cdots + y_d.
\end{equation}
The inverse of $\gamma$, henceforth denoted $\fd : \mathbb R_{\geq 0}^{d;\downarrow} \to \mathbb R_{\geq 0}^{d}$, takes finite differences; specifically, $\fd$ maps $x = (x_1, \ldots, x_{d}) \in \mathbb R_{\geq 0}^{d;\downarrow}$ to $\fd(x) = (\fd_1(x), \ldots, \fd_d(x))$, where
\begin{equation}
    \label{eq:diffs}
    \fd_i(x) = \fd_i(x_1, \ldots, x_{k}) = \begin{cases} x_i - x_{i+1} & 1 \leq i < d \\ x_d & i = k \end{cases}.
\end{equation}
\begin{defn}
    For each $x = (x_1, \ldots, x_d) \in \mathbb R_{\geq 0}^{d}$, the \textit{size of $x$}, $\abs{x}$, is the sum of its elements
    \begin{equation}
        \abs{x} = x_1 + \cdots + x_d,
    \end{equation}
    The \textit{normalization of $x$} is defined\footnote{Assuming $x \in \mathbb R_{\geq 0}^{d}$ is not equal to all-zero $d$-tuple, $x \neq (0, \ldots, 0)$, so that $\abs{x} > 0$.} as
    \begin{equation}
        \frac{x}{\abs{x}} = \left(\frac{x_1}{\abs{x}}, \ldots, \frac{x_d}{\abs{x}}\right).
    \end{equation}
\end{defn}
Two subsets of $\mathbb R_{\geq 0}^{d;\downarrow}$ will be crucial to the results of \cref{sec:keyl_exclusive}. The first subset was already discussed in \cref{sec:schur_weyl}, namely partitions of $n$ with length at most $d$: $\yf^d_n \subseteq \mathbb R_{\geq 0}^{d;\downarrow}$. If the length, $\ell$, of $\lambda = (\lambda_1, \ldots, \lambda_\ell)$ is strictly less than $d$, then it can be viewed as a element of  $\mathbb R_{\geq 0}^{d;\downarrow}$ by padding $\lambda$ with $d-\ell$ zeros, i.e. $\lambda \cong (\lambda_1, \ldots, \lambda_\ell, 0, \ldots, 0)$. The second subset corresponds to the set of possible eigenvalues, or spectra, of density operators $\dens(\mathbb C^d)$.
\begin{defn}
    The set of \textit{spectra}, or sorted probability distributions, is
    \begin{equation}
        \spectra^{d} = \{ s \in \mathbb R_{\geq 0}^{d;\downarrow} \mid  {\sum}_{i=1}^{d} s_i = \abs{s} = 1 \}.
    \end{equation}
\end{defn}
While the normalization of any partition, $\lambda \in \yf^d_n$, is a spectrum, $\frac{\lambda}{n} \in \spectra^d$, multiplying a spectrum, $s \in \spectra^d$, by $n \in \mathbb N$ does not necessarily produce a partition because the entries of $ns$, $(ns_1, \ldots, ns_d)$, may not be integer-valued. Nevertheless, $ns \in \mathbb R_{\geq 0}^{d;\downarrow}$ can always be approximated by a partition, $\lambda \in \yf_n^d$, so that $\abs{\lambda_i - ns_i} \leq 1$ for all $i \in \{1, \ldots, d\}$.\footnote{An explicit scheme for accomplishing such an approximation is to let $t = n - {\sum}_{i}\lfloor n s_i \rfloor$ and define $\lambda_i = \lfloor n s_i \rfloor + 1$ whenever $i \leq t$ and $\lambda_i = \lfloor n s_i \rfloor$ whenever $i > t$.} In \cref{sec:keyl_exclusive}, it will be useful to consider approximating $n s$ with a partition, $\lambda$, in a different manner, where i) degeneracies of $s$ are preserved, i.e., $\delta_i(s) = 0 \implies \delta_i(\lambda) = 0$, and ii) non-degeneracies of $s$ are adequately represented, e.g., $\delta_i(\lambda) \geq \delta_i(ns)$. The next lemma shows that this can always be accomplished by partitions, $\lambda$, whose size is approximately $n$.

\begin{prop}
    \label{prop:crit_approx}
    Let $s = (s_1, \ldots, s_k) \in \spectra_{d}$ be a spectrum and $n \in \mathbb N$. Let $\lambda = (\lambda_1, \ldots, \lambda_d) \in \mathbb N_{\geq 0}^{d;\downarrow}$ be defined by
    \begin{equation}
        \lambda_i = \lceil n (s_i - s_{i+1}) \rceil + \cdots + \lceil n (s_{d-1} - s_{d}) \rceil + \lceil n s_d \rceil,
    \end{equation}
    such that $\fd_i(\lambda) = \lceil \fd_i(ns) \rceil$ holds. Then $\lambda$ is a partition of size $\abs{\lambda}$ where\footnote{These bounds are also tight for every $d$: if $s = \tbinom{d+1}{2}^{-1}(d, d-1, \ldots, 1)$, then $n = 1$ or $n = \tbinom{d+1}{2}$ yields $\lambda = (d, d-1, \ldots, 1)$ with size $\abs{\lambda} = \tbinom{d+1}{2}$ which achieves the upper bound when $n = 1$ and the lower bound when $n = \tbinom{d+1}{2}$.}
    \begin{equation}
        \label{eq:crit_approx_3}
        n \leq \abs{\lambda} \leq n + \tbinom{d+1}{2} - 1.
    \end{equation}
\end{prop}
\begin{proof}
    First note that for all $1 \leq i \leq d$,
    \begin{equation}
        \epsilon_i \coloneqq \lceil \fd_i(ns) \rceil - \fd_i(ns),
    \end{equation}
    is upper and lower bounded by $0 \leq \epsilon_i < 1$.
    From this observation, it will be shown that $\lambda$ approximates $ns$, specifically,
    \begin{equation}
        \label{eq:crit_approx_2}
        0 \leq \lambda_i - n s_i < d - i + 1,
    \end{equation}
    To prove \cref{eq:crit_approx_2}, we use (reverse) induction starting from the base case of $i = d$.
    Since $\lambda_d = \fd_d(\lambda) = \lceil \fd_d(ns) \rceil = \lceil n s_d \rceil = n s_d + \epsilon_d$, we have $\lambda_d - n s_d = \epsilon_d$ and thus \cref{eq:crit_approx_2} holds when $i = d$. Then, assuming \cref{eq:crit_approx_2} holds for $i = j+1$, we prove it holds for $i = j$. Since
    \begin{align}
        \fd_j(\lambda)
        &= \lambda_{j} - \lambda_{j+1} = \lceil \fd_j(ns) \rceil \\
        &= \fd_j(ns) + \epsilon_j = n s_j - n s_{j+1} + \epsilon_j,
    \end{align}
    we conclude that $\lambda_{j} - n s_j = \lambda_{j+1} - n s_{j+1} + \epsilon_j$ and thus $0 \leq \lambda_{j} - ns_{j} < d - j + 1$ which is \cref{eq:crit_approx_2} for $i = j$.
    Finally, \cref{eq:crit_approx_3} follows from \cref{eq:crit_approx_2} by summing over all $i$:
    \begin{equation}
        0 \leq \abs{\lambda} - n \abs{s} < \sum_{i=1}^{d} (d-i+1) = \tbinom{d+1}{2}.
    \end{equation}
    Since $\abs{\lambda}$ is necessarily an integer, $\abs{s} = 1$, and the upper bound above is strict, \cref{eq:crit_approx_2} holds.
\end{proof}

\section{Proof of \cref{lem:de_finetti_compatible}}
\label{sec:proof_finetti_lemma}

\begin{proof}[Proof of \cref{lem:de_finetti_compatible}]
    Let $d_{J} = \dim(H_J)$, let $\mu_{J}$ be the $U(d_{J})$-invariant Haar probability measure over the space of pure states $\proj(H_J)$,
    For any $k \in \mathbb N$, the orthogonal projection operator, $\Pi_J^{(k)}$, onto the symmetric subspace, $\symsub^{k} H_J$, is proportional to the expected value of $\psi_J^{\otimes k}$ when $\psi_J$ is sampled according to the probability measure $\mu_{J}$:
    \begin{equation}
        \label{eq:haar_sym}
        \Pi_{J}^{(k)} = \tbinom{k+d_{J}-1}{k}\int_{\proj(H_{J})} \mu_{J}(\diff \psi_{J}) \psi_{J}^{\otimes k},
    \end{equation}
    where the normalization factor is simply $\Tr[\Pi_{J}^{(k)}] = \tbinom{k+d_{J}-1}{k}$. The proof of \cref{eq:haar_sym} follows from Schur's lemma (see \cite[Proposition 6]{harrow2013church}).

    Next, define the map $\tau_{\s M} : \proj(H_J) \to \dens(H_{\s M})$ by
    \begin{equation}
        \tau_{\s M}(\psi_J) = \Tr_{mJ\setminus\s M}(\psi_J^{\otimes m}).
    \end{equation}
    Let $\nu_{\s M}$ be the push-forward measure of $\mu_{J}$ through $\tau_{\s M}$, i.e. $\nu_{\s M} = \mu_{J} \circ \tau_{\s M}^{-1}$.

    Next note that the coefficients $\tau_{\s M}(\psi_J)$ are homogeneous polynomials of degree $m$ in the coefficients of $\psi_{J}$, and thus $\tau_{\s M}$ is continuous and measurable. Additionally, by construction, the image of $\tau_{\s M}$ is precisely the set of compatible $\s M$-product states $\s C_{\s M}$. Therefore, since $\proj(H_J)$ is compact (as $H_J$ is finite-dimensional), $\s C_{\s M}$ is also compact (and thus closed).
    Moreover, the support of the pushforward measure, $\nu_{\s M} = \mu_J \circ \tau_{\s M}^{-1}$ is equal to $\s C_{\s M}$.\footnote{This is because, by the closure of $\s C_{\s M}$, $\rho_{\s M} \not \in \s C_{\s M}$ implies there exists an open set, $O$ containing $\rho_{\s M}$, such that $O \cap \s C_{\s M} = \emptyset$ which implies $\nu_{\s M}(O) = \mu_J(\tau_{\s M}^{-1}(O)) = \mu_J(\emptyset) = 0$, i.e. $\rho_{\s M}$ is not in support of $\nu_{\s M}$. Moreover, if $\sigma_{\s M} \in \s C_{\s M}$ and $O'$ is any open set containing $\sigma_{\s M}$, $\tau_{\s M}^{-1}(O')$ is non-empty and open (by continuity of $\tau_{\s M}$) in $\proj(H_J)$ and thus $\nu_{\s M}(O') = \mu_J(\tau_{\s M}^{-1}(O')) > 0$. Therefore, because $\nu_{\s M}(O') > 0$ for all open sets containing $\sigma_{\s M}$, $\sigma_{\s M}$ is in the support $\nu_{\s M}$.}
    Finally, using \cref{eq:haar_sym}, linearity of $\Tr_{mJ\setminus\s M}$, and a change of variables,
    \begin{align}
        &\Tr_{mJ\setminus\s M}^{\otimes n}(\Pi^{(nm)}_J) \nonumber \\
        &\quad\propto \int_{\proj(H_J)} \mu_{J}(\diff \psi_J) (\Tr_{mJ\setminus\s M}(\psi_J^{\otimes m}))^{\otimes n}, \\
        &\quad= \int_{\proj(H_J)} \mu_{J}(\diff \psi_J) (\tau(\psi_J))^{\otimes n}, \\
        &\quad= \int_{\s C_{\s M}} \nu_{\s M}(\diff \sigma_{\s M}) \sigma_{\s M}^{\otimes n}.
    \end{align}
\end{proof}

\section{Pure vs. Full QMP}
\label{sec:pure_vs_mixed_qmp}
One might wonder why the version of the QMP considered in this paper (\cref{prob:qmp}) seems to be exclusively interested in the existence of joint states that are \textit{pure}, $\psi_J \in \proj(H_{J})$, instead of the more general \textit{density operator}, $\rho_{J} \in \dens(H_{J})$. In order to distinguish between these two types of QMP, the former is sometimes called the \textit{pure} QMP, while the latter is sometimes called the \textit{full} QMP. Of these two variants, the full QMP is arguably a much closer analogy to the classical marginals problem~\cite{fritz2012entropic}.

The distinction between these two variants is strongest when the marginal scenario under consideration, $\s M = (S_1, \ldots, S_m)$, has disjoint marginal contexts, i.e. $S_i \cap S_j = \emptyset$ for all $i \neq j$, or equivalently $H_{\s M} = H_{S_1}\otimes \cdots \otimes H_{S_m} \cong H_J$. Under the assumption of disjoint marginal contexts, the full QMP becomes trivial; every collection of density operators $(\rho_{S_1}, \ldots, \rho_{S_m})$ are the $\s M$-marginals of the density operator $\rho_J = \rho_{S_1} \otimes \cdots \otimes \rho_{S_m}$. On the other hand, under this assumption, the pure QMP remains non-trivial.

At the level of generality considered in this paper, wherein the marginal scenario $\s M = (S_1, \ldots, S_m)$ is permitted to contain overlapping marginal contexts, e.g., $S_i \cap S_j \neq \emptyset$, the distinction becomes less important because the marginals of a mixed state can equivalently be viewed as the marginals of any of its purifications. Specifically, there exists a density operator $\sigma_{J} \in \dens(H_{J})$ with marginals $\rho_{S} = \Tr_{J\setminus S}(\sigma_{J})$ for all $S \in \s M$ if and only if there exists a joint pure state $\psi_{JJ'} \in \proj(H_{J} \otimes H_{J'})$ (where $H_{J'} \cong H_{J}$) such that $(\Tr_{J\setminus S}\otimes \Tr_{J'})(\sigma_{J}) = \rho_{S}$ for all $S \in \s M$. Consequently, the techniques developed in this paper, which directly apply to the pure QMP, can also be applied to any instance of the full QMP without substantial modification.

\section{Diagrammatics}
\label{sec:diagrammatics}

The purpose of this section is to briefly introduce some diagrammatic notation that will be useful for performing a few calculations in \cref{sec:n1}. The particular notations involving symmetrization and antisymmetrization used here (\cref{eq:sym_2_X} and onward), are taken from \citeauthor{cvitanovic2008group}'s excellent textbook~\cite{cvitanovic2008group} on diagrammatic calculations of invariants of Lie groups, and are essentially the same those used by \citeauthor{penrose1971applications}~\cite{penrose1971applications}. For a categorical justification of this notation, see~\cite{selinger2012finite}. For further applications within quantum theory, see~\cite{coecke2010compositional,wood2011tensor,biamonte2017tensor}.

The essential idea is to depict linear operators, $L : H_X \to H_Y$, by pictures with corresponding inputs and outputs:
\begin{equation}
    \strdia{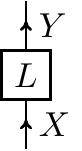}.
\end{equation}
Of course, two linear operators can be combined in at least three different ways; specifically, by addition $+$, tensor product $\otimes$, and sequential composition $\circ$. These operations are depicted respectively as
\begin{align}
    \strdia{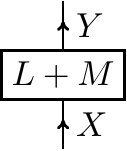} &= \strdia{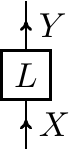} + \strdia{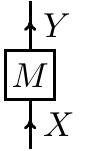},\\
    \strdia{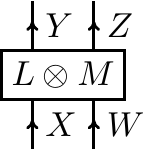} &= \strdia{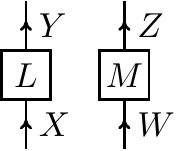}, \\
    \strdia{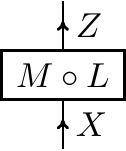} &= \strdia{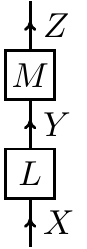}.
\end{align}
Important special cases of this notation include the identity operator $\ident_X : H_X \to H_X$,
\begin{equation}
    \strdia{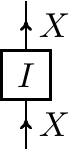} = \strdia{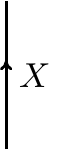},
\end{equation}
and vectors $\ket{\psi_X} : \mathbb C \to  H_{X}$ (and their conjugates $\bra{\psi_X} : H_{X} \to \mathbb C$) as
\begin{equation}
    \strdia{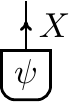}, \quad \text{and} \quad \strdia{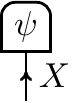}.
\end{equation}
This notation is especially elegant for depicting two concepts frequently encountered in this paper: the partial trace and permutations.

First, given a bipartite operator $L : H_{X} \otimes H_{Y} \to H_{X} \otimes H_{Y}$, the partial trace $\Tr_{Y}$ over $Y$, is depicted as
\begin{equation}
    \strdia{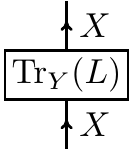} = \strdia{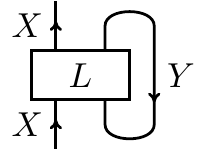}.
\end{equation}
The trace over the identity operator $\ident_{X} : H_{X} \to H_{X}$, which is equal to the dimension of $H_X$, is therefore depicted as a closed loop
\begin{equation}
    d_X = \dim(H_X) = \strdia{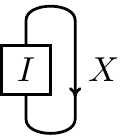} = \strdia{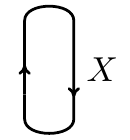} = \strdia{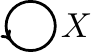}.
\end{equation}
Second, the tensor permutation representation, $T_{X} : \sym_{m} \to \s L(\s H_{X}^{\otimes m})$, of the symmetric group, $\sym_{m}$, has elements depicted naturally as follows. When $m=2$, $\sym_{2} = \{ e, (12) \}$, and the identity $T_{X}(e)$ and swap $T_{X}((12))$ are depicted respectively by
\begin{equation}
    \strdia{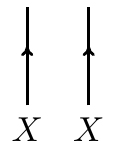}, \quad \text{and} \quad \strdia{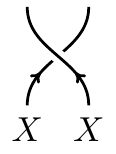}.
\end{equation}
Analogously, for $m=3$, the $3!=6$ permutations in $\sym_{3}$ are depicted by
\begin{align}
    \begin{split}
        \strdia{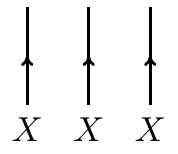}, &\qquad \strdia{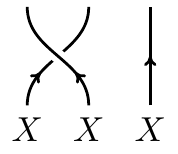}, \qquad \strdia{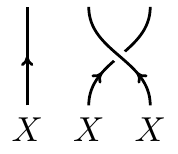},\\
        \strdia{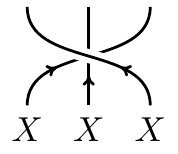}, &\qquad \strdia{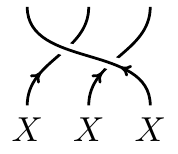}, \qquad \strdia{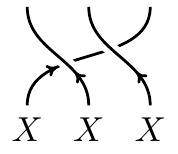}.
    \end{split}
\end{align}
The orthogonal projection operator, $\Pi^{(2)}_{X}$, onto the symmetric subspace $\symsub^{2}\s H_{X} \subseteq \s H^{\otimes 2}_X$, referred to as the \textit{symmetrization operator}, is given the following unique notation:
\begin{equation}
    \label{eq:sym_2_X}
    \strdia{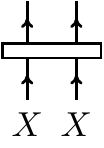}
    \coloneqq
    \strdia{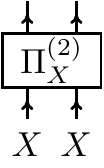}
    =
    \frac{1}{2}\strdia{S2_e_X}
    +
    \frac{1}{2}\strdia{S2_12_X}.
\end{equation}
Similarly, the orthogonal projection operator onto the antisymmetric subspace $\antisymsub^{2}\s H_{X} \subseteq \s H^{\otimes 2}$, denoted by $\Pi_{X}^{(1,1)}$ and referred to as the \textit{antisymmetrization operator}, is depicted in a complementary manner:
\begin{equation}
    \strdia{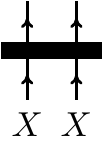}
    \coloneqq
    \strdia{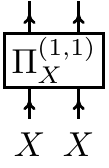}
    =
    \frac{1}{2}\strdia{S2_e_X}
    -
    \frac{1}{2}\strdia{S2_12_X}.
\end{equation}
Generalizing this notation for the orthogonal projection operators onto the symmetric and antisymmetric subspaces of $\s H_X^{\otimes m}$ for $m > 2$ can be done recursively as follows. For the sake of clarity, the Hilbert space label, $X$, can often be omitted without introducing ambiguity.
\begin{align}
    \strdia{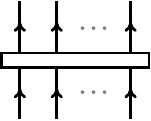}
    &=
    \frac{1}{m} \left(
        \strdia{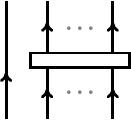}
        + (m-1)
        \strdia{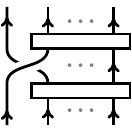}
    \right),\\
    \strdia{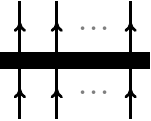}
    &=
    \frac{1}{m} \left(
        \strdia{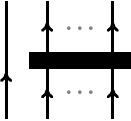}
        - (m-1)
        \strdia{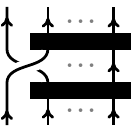}
    \right).
\end{align}
Finally, in order to generalize the above symmetrization and antisymmetrization notation to the case of multipartite Hilbert spaces, e.g., $H_{XY} = H_X \otimes H_Y$, we introduce the following notational definition for the joint symmetrization $\Pi^{(2)}_{XY}$:
\begin{equation}
    \label{eq:parallel_sym_XY}
    \strdia{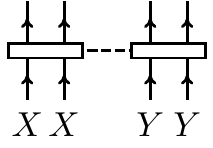}
    =
    \frac{1}{2}
    \strdia{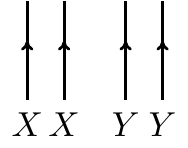}
    +
    \frac{1}{2}
    \strdia{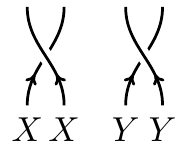}.
\end{equation}

\section{The \texorpdfstring{$\boldsymbol{n=1}$}{n=1} Case}
\label{sec:n1}

This section explores the strength of the constraint imposed by \cref{eq:nth_order} for the special case when $n=1$ (equivalently \cref{eq:n1}) for the purposes of detecting incompatible $\s M$-product states. For marginal scenarios involving disjoint marginal contexts, it will be shown that \cref{eq:n1} happens to be satisfied by \textit{all} $\s M$-product states, and therefore is useless for the QMP. For at least some marginal scenarios involving non-disjoint marginal contexts, it will shown that \cref{eq:n1} is already capable of witnessing the incompatibility of certain $\s M$-product states. Finally, it is shown that for some (admittedly degenerate) marginal scenarios, the constraint imposed by \cref{eq:n1} is also sufficient for the corresponding QMP.

Throughout this section, unipartite subsystems are labeled alphabetically, e.g., $A, B, C\ldots$, and their respective dimensions denoted by lower-case letters, e.g., $a = d_A$, $b = d_B$, etc.

Consider the marginal scenario $\s M = (A, B)$ ($m=2$) for the joint context $J = AB$. In this scenario, the QMP is already fully solved: $\rho_{A}$ and $\rho_{B}$ are the marginals of some pure state $\psi_{AB}$ if and only if $\spec(\rho_{A}) = \spec(\rho_{B})$ (see \cref{sec:bipartite}). To what extent, if any, does \cref{eq:n1} reproduce this known solution? If $\rho_{A}$ and $\rho_{B}$ are the marginals of some pure state $\psi_{AB}$, \cref{eq:n1} implies
\begin{equation}
    \label{eq:n1_A_B}
    \rho_{A} \otimes \rho_{B} \leq (\Tr_{B}\otimes \Tr_{A})(\Pi^{(2)}_{AB}).
\end{equation}
To calculate the right-hand side of the above inequality, it will be convenient to use the diagrammatic notation introduced in \cref{sec:diagrammatics}. Specifically, $\Pi^{(2)}_{AB}$ can be depicted using \cref{eq:parallel_sym_XY} (with $X,Y$ substituted by $A,B$), and thus $(\Tr_{B}\otimes \Tr_{A})(\Pi^{(2)}_{AB})$ is equal to
\begin{align}
    \strdia{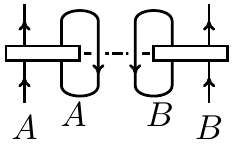}
    &=
    \frac{1}{2}
    \strdia{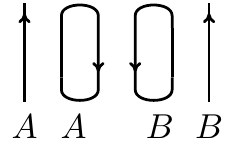}
    +
    \frac{1}{2}
    \strdia{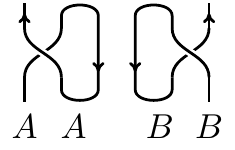},
    \\
    &=
    \frac{ab}{2}
    \strdia{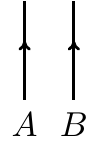}
    +
    \frac{1}{2}
    \strdia{parallel_AB_ident},\\
    &= \frac{1+ab}{2} \strdia{parallel_AB_ident}.
\end{align}
Therefore, \cref{eq:n1_A_B} is equivalent to $\rho_{A} \otimes \rho_{B} \leq \frac{1}{2}(1+ab) \ident_{A} \otimes \ident_{B}$, i.e.
\begin{equation}
    \strdia{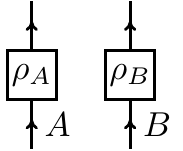} \leq \frac{1+ab}{2} \strdia{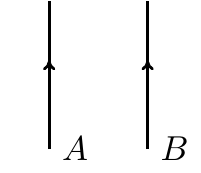},
\end{equation}
which is an inequality satisfied by all $(A, B)$-product states $\rho_{A} \otimes \rho_{B}$ because $\rho_X \leq \ident_X$ already holds for all $\rho_X \in \dens(H_X)$ and $ab \geq 1$. In fact, it is not too difficult to show that when $\s M = (X_1, \ldots, X_k)$ contains disjoint contexts, i.e. $X_i \cap X_j = \emptyset$ for $i \neq j$, the inequality in \cref{eq:n1} is always trivial because $\Tr_{mJ\setminus\s M}(T_{J}(\pi)) \geq \ident_{\s M}$ holds for all $\pi \in \sym_k$ and thus $\Tr_{mJ\setminus\s M}(\Pi^{(m)}_{J}) \geq \ident_{\s M}$ also. Fortunately, the same is not necessarily true for non-disjoint marginal scenarios.

For an example of a non-trivial instance of \cref{eq:n1}, consider the marginal scenario $\s M = (AB, AC, BC)$ for the joint context $J = ABC$. In this scenario, the aforementioned operator inequality becomes
\begin{equation}
    \label{eq:AB_AC_BC_ex}
    \rho_{AB} \otimes \rho_{AC} \otimes \rho_{BC} \leq (\Tr_{C} \otimes \Tr_{B} \otimes \Tr_{A})(\Pi_{ABC}^{(3)}).
\end{equation}
The projector $\Pi_{ABC}^{(3)}$ onto $\symsub^{3}(H_{A} \otimes H_{B} \otimes H_{C})$ can be expressed as
\begin{align}
    &\strdia{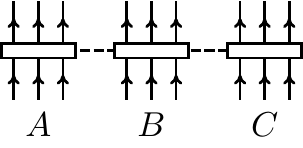}=
    \frac{1}{3!}
    \bigg[
    \strdia{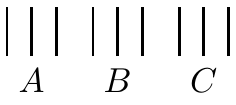}
    +\strdia{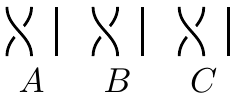}+ \\
    &\strdia{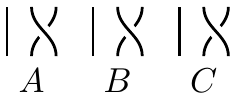}
    +\strdia{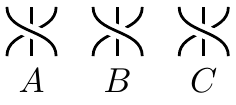}
    +\strdia{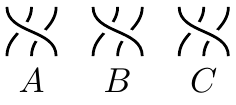}
    +\strdia{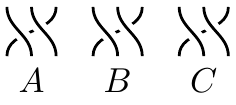}
    \bigg]. \nonumber
\end{align}
Therefore, $(\Tr_{C} \otimes \Tr_{B} \otimes \Tr_{A})(\Pi^{(3)}_{ABC})$ becomes
\begin{align}
    &\strdia{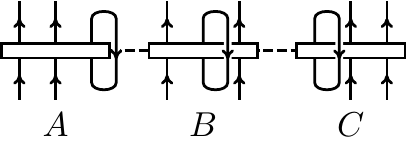} =
    \frac{1}{3!}
    \bigg[
    abc\strdia{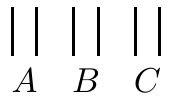} + \\
    &+c\strdia{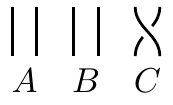}
    +a\strdia{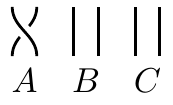}
    +b\strdia{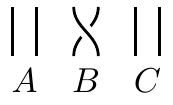}
    +2\strdia{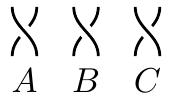}
    \bigg]. \nonumber
\end{align}
To show that \cref{eq:AB_AC_BC_ex} is a non-trivial constraint, we consider consider the case of three qubits, i.e. $a=b=c=2$. For a given pair of qubits, the unique antisymmetric pure state (also called the singlet state), $\Phi = \ket{\Phi}\bra{\Phi} \in \proj(\mathbb C^2 \otimes \mathbb C^2)$, can be identified with $\ket{\Phi} = \frac{1}{\sqrt{2}} (\ket{0 1} - \ket{1 0})$ and depicted as follows
\begin{equation}
    \strdia{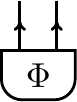} = \strdia{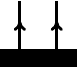}, \quad \text{s.t.} \quad \strdia{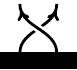} = - \strdia{levi_civita_2}, \quad \strdia{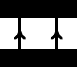} = 1.
\end{equation}
When applied to $(\Tr_{C} \otimes \Tr_{B} \otimes \Tr_{A})(\Pi^{(3)}_{ABC})$ (assuming $a = b = c = 2$) we obtain the identity
\begin{equation}
    \strdia{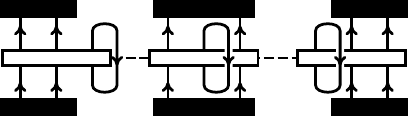} = \frac{1}{3!}(abc-a-b-c-2) = 0,
\end{equation}
which, when combined with \cref{eq:AB_AC_BC_ex} proves that the $(AB, AC, BC)$ marginals of a three-qubit pure state $\psi_{ABC}$ always satisfy
\begin{equation}
    \label{eq:AB_AC_BC_equality}
    \strdia{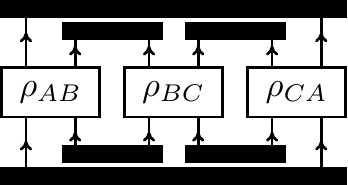} = 0.
\end{equation}
An example of an incompatible triple of states $(\rho_{AB}, \rho_{AC}, \rho_{BC})$ whose incompatibility is witnessed by the above equality constraint is the triple of anti-correlated states, $\rho_{AB} = \rho_{AC} = \rho_{BC} = \frac{1}{2}(\ket{01}\bra{01}+\ket{10}\bra{10})$, where the left-hand side evaluates to $2^{-5}$. Other examples includes the triple of singlets $\rho_{AB} = \rho_{AC} = \rho_{BC} = \Phi$ (with value $2^{-4}$), or the triple of maximally mixed states $\rho_{AB} = \rho_{AC} = \rho_{BC} = \frac{I}{2} \otimes \frac{I}{2}$ (with value $2^{-6}$). An example of an inconsistent triple of states for which \cref{eq:AB_AC_BC_equality} happens to be satisfied is $\rho_{AB} = \rho_{BC} = \ket{00}\bra{00}$ and $\rho_{AC} = \ket{11}\bra{11}$.

To conclude, consider the rather non-standard marginal scenario $\s M = (X, X)$ for the joint context $J = X$. Taken literally, the QMP for this marginal scenario is to determine, for any given pair of states $\rho_X, \sigma_X \in \dens(\s H_X)$, whether or not there exists a pure state $\psi_X \in \proj(\s H_X)$ such that $\rho_X = \psi_X$ and $\sigma_X = \psi_X$. This marginal scenario can be regarded as ``non-standard'' for at least two reasons: (i) the marginal context $X$ is repeated twice in $\s M$, and (ii) since $X$ is not a proper subset of $X$, $\rho_X$ and $\sigma_X$ are not proper marginals of $\psi_X$. Taken together, the QMP for this scenario has a simple solution: $\rho_X$ and $\sigma_X$ are compatible if and only if they are both pure states and equal to each other. Nevertheless, in this scenario \cref{eq:n1} is a valid constraint; in particular, it simplifies to $\rho_{X} \otimes \sigma_X \leq \Pi^{(2)}_{X}$, or diagrammatically
\begin{equation}
    \label{eq:XX_n1}
    \strdia{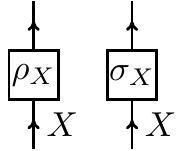} \leq \strdia{symmetrization_2_X}.
\end{equation}
The above inequality implies that $\Tr(\sigma_X \rho_X) = 1$ since
\begin{equation}
    0 \leq \strdia{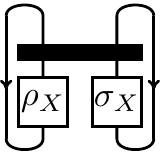} \leq \strdia{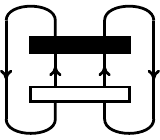} = 0,
\end{equation}
and
\begin{align}
    \strdia{tensor_rX_sX_antisym}
    &= \frac{1}{2} \strdia{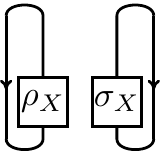} - \frac{1}{2}\strdia{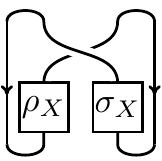}, \\
    &= \frac{1}{2} \left(1 - \strdia{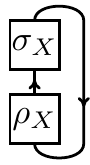}\right).
\end{align}
Since $\Tr(\sigma_X \rho_X) = 1$ holds if and only if $\sigma_X = \rho_X = \psi_X$ for some pure state $\psi_X$, we see that \cref{eq:n1}, which becomes \cref{eq:XX_n1}, is both necessary and \textit{sufficient} for the $\s M = (X, X)$ instance of the QMP.

\section{The Bipartite QMP}
\label{sec:bipartite}

This section considers the bipartite marginal scenario, $\s M = (A,B)$, for the joint context $J = AB$. For this scenario, the QMP is fully solved and admits of a simple solution: $\rho_{A}$ and $\rho_{B}$ are the marginals of a joint pure state $\psi_{AB} \in \proj(H_{AB})$ if and only if they have the same spectrum~\cite{tyc2015quantum,klyachko2004quantum}. A natural question arises: how does \cref{thm:main} recover this well-known result?

To answer this question, first let $a = \dim(H_A)$, $b = \dim(H_B)$, and let $\ell = \min(a,b)$. For this scenario, \cref{eq:nth_order} becomes
\begin{equation}
    \label{eq:bipartite_nth_order}
    (\rho_{A} \otimes \rho_{B})^{\otimes n} \leq (\Tr_{B}\otimes\Tr_{A})^{\otimes n}(\Pi_{AB}^{(2n)}).
\end{equation}
Now let $s_{A} \in \spectra^{a}$ and $s_{B} \in \spectra^{b}$ be the spectra of $\rho_{A}$ and $\rho_{B}$. Using the results of \cref{sec:keyl_exclusive} (or essentially the spectral estimation theorem~\cite{keyl2005estimating,christandl2006spectra}), together with \cref{eq:haar_sym}, it is possible to show that the exponential factor in \cref{eq:incompatibility_measure}, $\Omega(\rho_{A}\otimes\rho_{B})$, depends only on $r_A$ and $r_B$ and is equal to:
\begin{align}
    \Omega(\rho_{A} \otimes \rho_{B}) = \inf_{r \in \spectra^{\ell}}(\kl{s_{A}}{r}+\kl{s_{B}}{r}),
\end{align}
where $\kl{p}{q}$ is the relative entropy $\kl{p}{q} = \sum_{i} p_i (\ln p_i - \ln q_i)$. Since $\kl{p}{q}$ only vanishes if $p = q$, $\Omega(\rho_{A} \otimes \rho_{B})$ only vanishes if $s_{A} = s_{B}$. Therefore, we conclude that $\rho_{A}$ and $\rho_{B}$ are the $(A,B)$-marginals of a pure state $\psi_{AB} \in \proj(\mathbb C^{a} \otimes \mathbb C^{b})$ if and only if they have equal spectra.
Additionally, using Pinsker's inequality~\cite{reid2009generalised}, $\norm{p-q}_1^{2} \leq 2\kl{p}{q}$, and the triangle inequality for $\norm{\cdot}_1$, we obtain:
\begin{align}
    \norm{s_A - s_B}_1^2
    &\leq 3(\norm{s_A - r}_1^2 + \norm{s_B - r}_1^2)\\
    &\leq 6 (\kl{s_A}{r}+\kl{s_B}{r}).
\end{align}
Therefore, $\Omega(\rho_{A} \otimes \rho_{B}) \geq \norm{s_A - s_B}_1^2/6$.

A more direct consequence of \cref{eq:bipartite_nth_order} is the following proposition.
\begin{prop}
    Let $n \in \mathbb N$ and let $\alpha \in \yf_{n}^{a}$ and $\beta \in \yf_{n}^{b}$ be partitions. If $\rho_{A}\otimes \rho_{B}$ satisfies \cref{eq:bipartite_nth_order}, then
    \begin{equation}
        \label{eq:bipartite_result}
        s_{\alpha}(r_{A})s_{\beta}(r_{B}) \leq \sum_{\lambda \in \yf_{2n}^{\ell}} c^{\lambda}_{\alpha\beta} \frac{\dim(\schur{\lambda}^{a})\dim(\schur{\lambda}^{b})}{\dim(\specht{\lambda})},
    \end{equation}
    where $r_{A}$ and $r_{B}$ are the spectra of $\rho_{A}$ and $\rho_{B}$ respectively. Additionally, $s_{\alpha}$ and $s_{\beta}$ are Schur functions~\cite{sra2016inequalities,sagan2013symmetric} and $c^{\lambda}_{\alpha\beta}$ is the Littlewood-Richardson coefficient~\cite{littlewood1934group,fulton2000eigenvalues,pak2019largest}.
\end{prop}
\begin{proof}
    One of the most powerful tools for decomposing bipartite Hilbert spaces, specifically the symmetric subspace of a bipartite system $\symsub^{k}(H_{A} \otimes H_{B}) \cong V^{ab}_{(k)}$, is known as $\mathrm{GL}(a)\times\mathrm{GL}(b)$-duality~\cite{howe1987gl} (see also~\cite[Eq. (2.25)]{walter2014multipartite}):
    \begin{equation}
        V^{ab}_{(k)} \cong \bigoplus_{\lambda \in \yf_{k}^{\ell}} \schur{\lambda}^{a} \otimes \schur{\lambda}^{b},
    \end{equation}
    where $\ell = \min(a,b)$.
    Using this result, and applying $\Pi_{A}^{\alpha} \otimes \Pi_{B}^{\beta}$ to the right-hand-side of \cref{eq:nth_order}, we obtain
    \begin{align}
        \begin{split}
            &\Tr_{AB}^{\otimes n}\{(\Pi_{A}^{\alpha} \otimes \Pi_{B}^{\beta})(\Tr_{B}^{\otimes n}\otimes\Tr_{A}^{\otimes n})(\Pi_{AB}^{(2n)})\}\\
            %&\quad=\Tr_{AB}^{\otimes 2n}\{((\Pi_{A}^{\alpha} \otimes \ident_{B}^{\otimes n}) \otimes (\ident_{A}^{\otimes n} \otimes \Pi_{B}^{\beta}))\Pi_{AB}^{(2n)}\},\\
            %&\quad=\Tr_{AB}^{\otimes 2n}\{((\Pi_{A}^{\alpha} \otimes \ident_{B}^{\otimes n}) \otimes (\Pi_{A}^{\beta} \otimes \ident_{B}^{\otimes n})\Pi_{AB}^{(2n)}\},\\
            %&\quad=\Tr_{A}^{\otimes 2n}\{(\Pi_{A}^{\alpha} \otimes \Pi_{A}^{\beta})\Tr_{B}^{\otimes 2n}(\Pi_{AB}^{(2n)})\}, \\
            %&\quad=\sum_{\lambda \in \yf^{\ell}_{2n}}\Tr_{A}^{\otimes 2n}\{(\Pi_{A}^{\alpha} \otimes \Pi_{A}^{\beta})\Pi_{A}^{\lambda}\}\frac{\dim(\schur{\lambda}^{b})}{\dim(\specht{\lambda})},\\
            &\quad=\dim(\specht{\alpha}) \dim(\specht{\beta}) \hspace{-0.5em}\sum_{\lambda \in \yf^{\ell}_{2n}}\hspace{-0.5em}c^{\lambda}_{\alpha\beta}\frac{\dim(\schur{\lambda}^{a})\dim(\schur{\lambda}^{b})}{\dim(\specht{\lambda})},
        \end{split}
    \end{align}
    where $c^{\lambda}_{\alpha\beta}$ counts the multiplicity of the $\sym_n\times \sym_n$ irreducible representation space $\specht{\alpha} \otimes \specht{\beta}$ inside $\specht{\lambda}$ under the restriction of $\sym_{2n}$ to $\sym_{n} \times \sym_{n}$. By comparison, applying $\Pi_{A}^{\alpha} \otimes \Pi_{B}^{\beta}$ to the left-hand-side of \cref{eq:nth_order} yields
    \begin{align}
        &\Tr((\Pi_{A}^{\alpha} \otimes \Pi_{B}^{\beta})(\rho_{A}^{\otimes n} \otimes \rho_{B}^{\otimes n}))\\
        &\quad=\Tr(\Pi_{A}^{\alpha}\rho_{A}^{\otimes n})\Tr(\Pi_{B}^{\beta}\rho_{B}^{\otimes n}),\\
        &\quad= s_{\alpha}(r_{A})\dim(\specht{\alpha})s_{\beta}(r_{B})\dim(\specht{\beta}).
    \end{align}
    Therefore, \cref{eq:bipartite_nth_order} implies \cref{eq:bipartite_result} and thus the claim holds.
\end{proof}

\section{Fermionic \& Bosonic QMP}
\label{sec:ferm_bos}

Our sufficient family of necessary inequality constraints can be modified to handle the fermionic and bosonic variants of the QMP. Recall that a state describing a system of $p$ fermions (resp. bosons) with $f$ internal degrees of freedom, is typically modeled by an element of the antisymmetric subspace $\antisymsub^{p} \mathbb C^{f}$ (resp. the symmetric subspace $\symsub^{p}\mathbb C^{f}$). Since $\antisymsub^{p}\mathbb C^{f}$ (resp. $\symsub^{p}\mathbb C^{f}$) can be viewed as a subspace of a $p$-partite composite Hilbert space $(\mathbb C^f)^{\otimes p}$, and $\symsub^{n}\antisymsub^{p}\mathbb C^{f}$ (resp. $\symsub^{n}\symsub^{p}\mathbb C^{f}$) serves as the respresentation space for an irreducible representation of $SU(\antisymsub^{p}\mathbb C^{f}) \cong SU(\dim(\antisymsub^{p}\mathbb C^{f})) \cong SU(\tbinom{f}{p})$ (resp. $SU(\tbinom{p+f-1}{p})$), the analogue of \cref{eq:haar_sym} holds and thus an analogue of \cref{lem:de_finetti_compatible} also holds. Altogether, a generalization of \cref{thm:main} holds:
\begin{cor}
    Let $H_{V} \subseteq H_{J}$ be a subspace of a joint Hilbert space, $H_{J}$. An $\s M$-product state, $\rho_{\s M} = \rho_{S_1} \otimes \cdots \otimes \rho_{S_m}$, is compatible with a joint pure state $\psi_{V} \in H_V \subseteq H_{J}$ in the subspace $H_{V}$ if and only if for all $n \in \mathbb N$,
    \begin{equation}
        \rho_{\s M}^{\otimes n} \leq \Tr_{m J\setminus \s M}^{\otimes n}(\Pi^{(nm)}_{V}),
    \end{equation}
    where $\Pi^{(nm)}_{V}$ is the projection operator onto the $nm$-symmetric subspace $\symsub^{nm}H_V \subseteq H_J^{\otimes nm}$.
\end{cor}

\section{Counting Solutions to the QMP}
\label{sec:ortho_sol}

Whenever a given $\s M$-product state, $\rho_{\s M} = \rho_{S_1}\otimes \cdots \otimes \rho_{S_m}$, is shown to be compatible, a natural follow-up problem is to determine whether or not the joint state, $\psi_{J}$, satisfying \cref{eq:cqmp} is unique. For the bipartite marginal scenario, $\s M = (\s A, \s B)$, if the common spectrum of $\rho_{A}$ and $\rho_{B}$ is $s = (s_1, \ldots, s_r, 0, \ldots, 0)$, with positive values distinct, i.e. $s_1 > \cdots > s_r > 0$, then the \textit{unique} pure state, $\psi_{AB} \in \proj(H_J)$, satisfying \cref{eq:cqmp} is $\psi_{AB} = \ket{\psi_{AB}}\bra{\psi_{AB}}$ where
\begin{equation}
    \ket{\psi_{AB}} = \sum_{i=1}^{r} \sqrt{s_i} \ket{\phi_{A}^{(i)}} \otimes \ket{\phi_{B}^{(j)}},
\end{equation}
where $\{\phi_{A}^{(i)}\}_{i=1}^{r}$ and $\{\phi_{B}^{(i)}\}_{i=1}^{r}$ are the eigenvectors of $\rho_{A}$ and $\rho_{B}$. If, however, the common spectrum, $(s_1, \ldots, s_r, 0, \ldots, 0)$, is degenerate in the sense that some of its values are identical, then the solution to \cref{eq:cqmp} may not be unique. A familiar example of this phenomenon, for the two-qubit Hilbert space $H_J \cong \mathbb C^{2} \otimes \mathbb C^{2}$, are the four Bell states all sharing the same pair of maximally-mixed, single-qubit marginals, $(\frac{\ident}{2}, \frac{\ident}{2})$.

The following result generalizes \cref{eq:nth_disco} by considering the possibility that an $\s M$-product may be compatible with multiple, orthogonal, joint states.
\begin{cor}
    Let $\rho_{\s M} = \rho_{S_1} \otimes \cdots \otimes \rho_{S_m}$ be an $\s M$-product state and $\{ \psi_{J}^{(1)}, \ldots, \psi_{J}^{(v)}\}$ be a set of joint pure states, satisfying i) for all $1 \leq j,k \leq v$,
    \begin{equation}
        \Tr(\psi_{J}^{(j)}\psi_{J}^{(k)}) = |\braket{\psi_{J}^{(j)}|{\psi_{J}^{(k)}}}|^{2} = \delta_{j,k},
    \end{equation}
    and ii) for all $1 \leq i \leq m$, and $1 \leq k \leq v$,
    \begin{equation}
        \rho_{S_i} = \Tr_{J\setminus S_i}(\psi_{J}^{(k)}).
    \end{equation}
    Then, the following inequality holds:
    \begin{equation}
        \label{eq:ortho_ineq}
        v^{nm}\rho_{\s M}^{\otimes n} \leq \sum_{\lambda \in \yf_{nm}^{v}}\Tr_{mJ\setminus\s M}^{\otimes n}(\Pi_{J}^{\lambda}).
    \end{equation}
\end{cor}
\begin{proof}
    Let $P_V$ be the orthogonal projection operator onto the subspace of $H_J$ spanned by $\{\psi_{J}^{(k)}\}_{k=1}^{v}$, i.e.,
    \begin{equation}
        P_V = \sum_{k=1}^{v} \psi_{J}^{(k)}.
    \end{equation}
    Since each $\psi_{J}^{(k)}$ has marginals $(\rho_{S_1}, \ldots, \rho_{S_m})$, we conclude
    \begin{equation}
        \label{eq:ortho_equality}
        \Tr_{mJ\setminus\s M}^{\otimes n}(P_V^{\otimes nm}) = v^{nm} \rho_{\s M}^{\otimes n}.
    \end{equation}
    Furthermore, $P_V^{\otimes nm}$ commutes with $T_J(\pi)$ for all $\pi \in \sym_{nm}$ and thus commutes with $\Pi_J^{\lambda}$ for every $\lambda \in \yf_{nm}^{d_{J}}$. In fact, we obtain
    \begin{equation}
        P_V^{\otimes nm} = \sum_{\lambda \in \yf_{nm}^{v}} \Pi_{J}^{\lambda}P_V^{\otimes nm} \Pi_{J}^{\lambda} \leq \sum_{\lambda \in \yf_{nm}^{v}}\Pi_{J}^{\lambda},
    \end{equation}
    because i) $\Pi_{J}^{\lambda} P_{V}^{\otimes nm} = 0$ for all $\lambda$ with length $\ell(\lambda) > v$, and ii) $P_V \leq \ident_{J}$. Applying $\Tr_{mJ\setminus \s M}^{\otimes n}$ yields \cref{eq:ortho_ineq}.
\end{proof}
Note that \cref{eq:ortho_ineq} is equivalent to \cref{eq:nth_disco} if $v = 1$. Also note that while \cref{eq:ortho_ineq} is necessary for the existence of $v$ orthogonal solutions to the QMP, satisfying \cref{eq:ortho_ineq} for all $n$ is generally insufficient (for $v > 1$) to conclude that $v$ orthogonal solutions to the QMP exist. For example, if $v = d_{J}$, \cref{eq:ortho_ineq} simplifies to
\begin{equation}
    \rho_{\s M}^{\otimes n} \leq \left(\frac{\ident_{\s M}}{d_{\s M}}\right)^{\otimes n}.
\end{equation}
The above constraint is evidently satisfied for all $n \in \mathbb N$, if and only if $\rho_{\s M}$ is the maximally-mixed $\s M$-product state, i.e., $\rho_{\s M} = \ident_{\s M}/d_{\s M}$. However, such states are generally incompatible~\cite{klyachko2002coherent}, e.g., it can be shown that $(\frac{\ident_{A}}{a}, \frac{\ident_{B}}{b})$ are not the $(A, B)$-marginals of \textit{any} pure state, $\psi_{AB}$, if $a \neq b$.

\section{Keyl Divergence \& State Discrimination}
\label{sec:keyl_exclusive}

The purpose of this section is to prove \cref{thm:constructive_state_discrim} which can be interpreted as an explicit strategy for asymmetric quantum state discrimination. While \cref{thm:constructive_state_discrim} is exclusively used by this paper in the proof of \cref{thm:main}, it may be of independent interest. Many of the results of this section come directly from \citeauthor{keyl2006quantum}'s work on a large-deviation-theoretic approach to quantum state estimation~\cite{keyl2006quantum}. The only additional insight not taken from \cite{keyl2006quantum} is the use of \cref{prop:crit_approx} in the proof of \cref{thm:constructive_state_discrim}. Both \cref{sec:schur_weyl} and \cref{sec:spectra_partitions} are considered prerequisites for this section.

\begin{defn}
    Let $x \in \mathbb R_{\geq 0}^{d;\downarrow}$ and let $\rho \in \dens(\mathbb C^{d})$, define
    \begin{equation}
        \label{eq:gpf}
        \Delta_{x}(\rho) = {\prod}_{i=1}^{d} \lpm_{i}(\rho)^{\fd_i(x)},
    \end{equation}
    where $\lpm_i(\rho)$ is the $i$th (leading) principal minor of $\rho$, i.e. the determinant of the upper-left $i\times i$-submatrix of $\rho$ with respect to some fixed, computational basis $\{ \ket{0}, \ldots, \ket{d}\}$.\footnote{If it happens that $\lpm_i(\rho) = 0$ and $\fd_i(x) = 0$ for some index $i$, then the indeterminant expression $0^0$ is taken to be equal to $1$.}
\end{defn}
The function defined in \cref{eq:gpf} is also referred to as the \textit{generalized power function}~\cite[Notation 4.1]{o2016efficient}. Note however, in \cite{o2016efficient}, $x$ is restricted to be a partition with length at most $d$, while $\rho$ is permitted to be any $d \times d$ complex-valued matrix.

\begin{prop}
    Let $s = (s_1, \ldots, s_d) \in \spectra^d$ be the spectrum of $\sigma \in \dens(\mathbb C^d)$. For all $x = (x_1, \ldots, x_d) \in \mathbb R_{\geq 0}^{k; \downarrow}$,
    \begin{equation}
        \Delta_{x}(\sigma) \leq \Delta_{x}(\diag(s_1, \ldots, s_d)) = {\prod}_{i=1}^{d} s_i^{x_i},
    \end{equation}
    with equality holding if and only if $\sigma = \diag(s_1,\ldots, s_d)$.
\end{prop}
\begin{proof}
    Consider any $i \in \{1, \ldots, d\}$. Let $(s^{(i)}_{1}, \ldots, s^{(i)}_i)$ with $s^{(i)}_1 \geq \cdots \geq s^{(i)}_i$ denote the eigenvalues of the $i\times i$ leading principal submatrix of $\sigma$ so that $\lpm_i(\sigma) = s^{(i)}_1 \cdots s^{(i)}_i$.
    According to Cauchy's interlacing theorem (see \cite{fisk2005very} or \cite[Thm. 4.3.17]{horn1985matrix} noting the reversed ordering of labels), for all $1 < i \leq d$,
    \begin{equation}
        s^{(i)}_{1} \geq s^{(i-1)}_1 \geq s^{(i)}_{2} \geq \cdots \geq s^{(i)}_{i-1} \geq s^{(i-1)}_{i-1} \geq s^{(i)}_{i}.
    \end{equation}
    Therefore, for any $k$ and $i$ such that $1 \leq k \leq i \leq d$, $s_k = s^{(d)}_k \geq s^{(i)}_k \geq s^{(k)}_k$. Therefore, for all $i \in \{1, \ldots, d\}$,
    \begin{equation}
        \label{eq:maximimal_principal_minor}
        \lpm_i(\sigma) \leq s_1 \cdots s_i,
    \end{equation}
    with equality holding (for all $i$) only if $\sigma = \diag(s_1,\ldots, s_d)$.
\end{proof}
\begin{prop}
    \label{prop:shannon}
    Let $s = (s_1, \ldots, s_d) \in \spectra^d \subseteq \mathbb R_{\geq 0}^{d;\downarrow}$. Then
    \begin{equation}
        \Delta_s(\mathrm{diag}(s)) = {\prod}_{i=1}^{d} s_i^{s_i} = \exp(- H(s)) > 0.
    \end{equation}
    where $H(s) = -{\sum}_{i=1}^{d} s_i \ln s_i$ is the Shannon entropy of $s$.
\end{prop}
\begin{cor}
    \label{cor:keyl_div}
    Let $\rho \in \dens(\mathbb C^{d})$ have spectrum $s = (s_1, \ldots, s_d) \in \spectra^d$ and let $U \in \Unit(d)$ be a unitary such that $\rho = U \diag(s_1, \ldots, s_d) U^{\dagger}$ and let $\sigma \in \dens(\mathbb C^{d})$. Then
    \begin{equation}
        \frac{\Delta_{s}(U^{\dagger} \sigma U)}{\Delta_{s}(\diag(s))} = \exp(-\keyl{\rho}{\sigma})
    \end{equation}
    where $\keyl{\rho}{\sigma}$ is defined as
    \begin{equation}
        \keyl{\rho}{\sigma} = {\sum}_{i=1}^{d} s_i \ln s_i - \fd_i(s) \ln \lpm_i(U^{\dagger} \sigma U),
    \end{equation}
    where $\keyl{\rho}{\sigma} \in [0, \infty]$ and $\keyl{\rho}{\sigma} = 0$ if and only if $\sigma = \rho$.
\end{cor}
Notice that if $\rho$ and $\sigma$ are simultaneously diagonalized by $U$ so that $U^{\dagger} \sigma U = \diag(t_1, \ldots, t_d)$, the quantity $\keyl{\rho}{\sigma}$ simplifies to the classical relative entropy, $\kl{s}{t} = \sum_{i=1}^{d} s_i (\ln s_i - \ln t_i)$,
also known as Kullback-Liebler divergence~\cite{kullback1997information}. Also note that, in general, $\keyl{\rho}{\sigma}$ does not equal the quantum relative entropy $\qrl{\rho}{\sigma} = \Tr(\rho (\ln \rho - \ln \sigma))$, but is nevertheless bounded by $\keyl{\rho}{\sigma} \leq \qrl{\rho}{\sigma}$~\cite{keyl2006quantum}. For these reasons, we refer to the quantity $\keyl{\rho}{\sigma}$ as \textit{Keyl-divergence}.

\begin{prop}
    \label{prop:hwv_lpm}
    Let $\lambda \in \yf_n^d$ be a partition of $n \in \mathbb N$ and let $\ket{\phi_\lambda^U} \in \schur{\lambda}^d$ be the twirled highest weight vector for $U \in \Unit(d)$ (see \cref{defn:twirled}). Then for all $\sigma \in \dens(\mathbb C^d)$,
    \begin{equation}
        \bra{\phi^U_\lambda} \tau_{\lambda}(\sigma^{\otimes \abs{\lambda}}) \ket{\phi^U_\lambda} = \Delta_{\lambda}(U^{\dagger}\sigma U).
    \end{equation}
\end{prop}
\begin{proof}
    This result is noted by \citeauthor{keyl2006quantum} as \cite[Eqs. (141) \& (151)]{keyl2006quantum} with reference to \cite[Section 49]{zhelobenko1973compact}.
\end{proof}
Henceforth, define the projection operator $\Phi_{\lambda}^{U} \in \s L(H^{\otimes \abs{\lambda}})$ by
\begin{equation}
    \Phi_{\lambda}^{U} \coloneqq \iota_{\lambda} (\ket{\phi_{\lambda}^{U}} \bra{\phi^{U}_{\lambda}} \otimes \ident_{\specht{\lambda}}) \iota_{\lambda}^{\dagger},
\end{equation}
so that,
\begin{equation}
    \label{eq:defn_of_big_phi}
    \Tr(\Phi_{\lambda}^{U} \rho^{\otimes \abs{\lambda}}) = \dim(\specht{\lambda})\Delta_{\lambda}(U^{\dagger} \rho U).
\end{equation}
\begin{cor}
    \label{cor:rational_state_discrim}
    Let $\rho, \sigma \in \dens(\mathbb C^d)$ and let $U \in \Unit(d)$ diagonalize $\rho$, i.e. $\rho = U\diag(s_1, \ldots, s_d)U^{\dagger}$. If $\rho$ has \textit{rational} spectra $s = (s_1, \ldots, s_d) \in \spectra^{d}$, i.e. there exists a $q \in \mathbb N$ such that $qs \in \yf_{q}^{d}$ is a partition of $q$, then for all $n \in \mathbb N$,
    \begin{equation}
        \label{eq:rational_state_discrim}
        \frac{\Tr(\Phi^{U}_{nqs} \sigma^{\otimes nq})}{\Tr(\Phi^{U}_{nqs}\rho^{\otimes nq})} = \exp (-nq\keyl{\rho}{\sigma}).
    \end{equation}
\end{cor}
\begin{proof}
    The proof follows from \cref{eq:defn_of_big_phi} and \cref{cor:keyl_div}. When \cref{eq:defn_of_big_phi} is applied to the numerator and denominator on the left-hand-side of \cref{eq:rational_state_discrim}, the common factor of $\dim(\specht{nqs}) > 0$ cancels out.
\end{proof}
A result similar to \cref{cor:rational_state_discrim} holds for arbitrary states $\rho \in \dens(\mathbb C^d)$, e.g., for states that do not have rational spectra.
\begin{thm}
    \label{thm:constructive_state_discrim}
    Let $\rho, \sigma \in \dens(\mathbb C^d)$, let $s = (s_1, \ldots, s_d) \in \spectra^{d}$ be the spectrum of $\rho$, and let $U \in \Unit(d)$ diagonalize $\rho$, such that $\rho = U\diag(s_1, \ldots, s_d)U^{\dagger}$.
    Then there exists a sequence, $n \mapsto \lambda^{n} \in \yf_{n}^{d}$ of partitions such that for all $n \in \mathbb N$,
    \begin{equation}
        \label{eq:constructive_state_discrim}
        \frac{\Tr(\Phi^{U}_{\lambda^{n}}\sigma^{\otimes n})}{\Tr(\Phi^{U}_{\lambda^{n}}\rho^{\otimes n})} \leq D(s)\exp (-(n-\tbinom{d+1}{2}+1) \keyl{\rho}{\sigma}).
    \end{equation}
    where $D(s)$ is a constant depending only on $s$.
\end{thm}
\begin{proof}
    The proof relies on an explicit construction of a sequence, $n \mapsto \lambda^{n}$, that satisfies the claim.
    For each non-negative integer $k \in \mathbb N_{\geq 0}$, let $\mu^{k}$ be the partition characterized by $\delta_i(\mu^k) = \lceil\fd_i(ks)\rceil$ (see \cref{prop:crit_approx}). Then for any state $\eta \in \dens(\mathbb C^d)$, we claim
    \begin{equation}
        \label{eq:eta_bounds}
        \Delta_{ks}(\eta) \Delta_{\mu^1}(\eta) \leq \Delta_{\mu^k}(\eta) \leq \Delta_{ks}(\eta).
    \end{equation}
    To see the upper bound, note that for all $i \in \{1, \ldots, d\}$, $\fd_i(\mu^k) \geq \fd_i(k s) > 0$ so $\lpm_i(\eta)^{\fd_{i}(\mu^k)} \leq \lpm_i(\eta)^{\fd_{i}(ks)}$ since $\lpm_i(\eta) < 1$ (noting \cref{eq:maximimal_principal_minor}).
    For the lower bound, note that $\fd_i(\mu^k) = \lceil \fd_i(ks) \rceil \leq \fd_i(ks) + \lceil \fd_i(s) \rceil = \fd_i(ks) + \fd_i(\mu^1)$, so $\Delta_{\mu^k}(\eta) \geq \Delta_{ks}(\eta) \Delta_{\mu^1}(\eta)$ holds.
    Next, apply \cref{cor:keyl_div} and \cref{eq:eta_bounds} (the upper bound when $\eta = U^{\dagger} \sigma U$, and the lower bound when $\eta = U^{\dagger} \rho U = \diag(s_1, \ldots s_d)$) to obtain
    \begin{equation}
        \label{eq:constructive_state_discrim_k}
        \frac{\Tr(\Phi_{\mu^{k}}^{U}\sigma^{\otimes \abs{\mu^k}})}{\Tr(\Phi_{\mu^{k}}^{U}\rho^{\otimes \abs{\mu^k}})} \leq \frac{\exp(-k \keyl{\rho}{\sigma})}{\Delta_{\mu^{1}}(\diag(s_1, \ldots s_d))}.
    \end{equation}
    Now, notice that \cref{eq:constructive_state_discrim_k} is almost in the form of \cref{eq:constructive_state_discrim}.
    The main obstacle remaining is simply that the size of $\mu^{k}$ needs to be decoupled from the spectra of $\rho$.
    Fortunately, \cref{prop:crit_approx} guarantees that $\abs{\mu^{k}}$ is, at least, approximately equal to $k$ because $k \leq \abs{\mu^{k}} \leq k + \tbinom{d+1}{2} - 1$. Moreover, since $\abs{\mu^{k+1}} \geq \abs{\mu^{k}}$, there always exists at least one value of $k$ such that $\mu^{k}$ has size approximately $n$ for any $n \in \mathbb N$; specifically, there exists a $k \in \mathbb N$ such that
    \begin{equation}
        \label{eq:optimal_k}
        n - \tbinom{d+1}{2} + 1 \leq \abs{\mu^{k}} \leq n.
    \end{equation}
    Now simply define $\lambda^{n} \in \yf_{n}^{d}$ by
    \begin{equation}
        \lambda^{n} = (\mu^{k}_1 + n - \abs{\mu^{k}}, \mu^{k}_{2}, \ldots, \mu_{d}^{k}),
    \end{equation}
    where $k$ is the largest such that $\mu_k$ satisfies \cref{eq:optimal_k}. Note that when $n$ is small ($n < \tbinom{d+1}{2} - 1$), is entirely possible for $k = 0$ and $\mu_0 = (0,0,\ldots,0)$, in which case, $\lambda^n = (n, 0, \ldots, 0)$.
    This definition ensures
    \begin{align}
        \Tr(\Phi^{U}_{\lambda^{n}} \rho^{\otimes n}) &= s_1^{n-\abs{\mu^{k}}} \Tr(\Phi^{U}_{\mu^k} \rho^{\otimes \abs{\mu^{k}}}), \\
        \Tr(\Phi^{U}_{\lambda^{n}} \sigma^{\otimes n}) &\leq \Tr(\Phi^{U}_{\mu^k} \sigma^{\otimes \abs{\mu^{k}}}).
    \end{align}
    Therefore, from \cref{eq:constructive_state_discrim_k,eq:optimal_k}, we conclude \cref{eq:constructive_state_discrim} where $D(s)$ is the constant
    \begin{align}
        D(s)
        &= s_1^{1-\tbinom{d+1}{2}} \left(\Delta_{\mu^{1}}(\diag(s_1, \ldots s_d))\right)^{-1} \\
        &= s_1^{1-\tbinom{d+1}{2}}\prod_{i=1}^{d} (s_1s_2 \cdots s_i)^{-\lceil \fd_i(s) \rceil}.
    \end{align}
\end{proof}
\begin{rem}
    \label{rem:state_discrim}
    To derive the inequality claimed in \cref{sec:necc_suff}, namely \cref{eq:incompatibility_measure}, from the result of \cref{thm:constructive_state_discrim}, note that $\tbinom{d+1}{2}-1 \leq d^2$ and substitute
    \begin{enumerate}[i)]
        \item $E_n = \Phi_{\lambda^{n}}^{U}$,
        \item $\infkeyl{\rho_{\s M}} = \inf_{\sigma_{\s M} \in \s C_{\s M}} \keyl{\rho_{\s M}}{\sigma_{\s M}}$,
        \item $c(\rho_{\s M}) = \ln D(\mathrm{spec}(\rho_{\s M}))$, and
        \item $d = d_{\s M} = \dim(H_{\s M})$.
    \end{enumerate}
    The claim that $\infkeyl{\rho_{\s M}}$ vanishes if and only if $\rho_{\s M} \in \s C_{\s M}$ follows from the compactness of $\s C_{\s M}$ and the following corollary.
\end{rem}
\begin{cor}
    \label{cor:keyl_div_from_R}
    Let $\s C \subseteq \dens(\mathbb C^d)$ be compact. Define
    \begin{equation}
        \infkeyl{\rho} \coloneqq \inf_{\sigma \in \s C}\keyl{\rho}{\sigma}.
    \end{equation}
    Then $\infkeyl{\rho} = 0$ if and only if $\rho \in \s C$.
\end{cor}
\begin{proof}
    If $\rho \in \s C$, then \cref{cor:keyl_div} implies $\infkeyl{\rho} = \keyl{\rho}{\rho} = 0$. Otherwise if $\rho \not \in \s C$, then consider, for each fixed $x \in \mathbb R_{\geq 0}^{d;\downarrow}$ and $U \in \Unit(d)$, that the function $\sigma \mapsto \Delta_x(U^{\dagger} \sigma U) \in [0,1]$ is continuous as $\lpm_i(U^{\dagger} \sigma U)$ is a polynomial in the coefficients of $\sigma$. The compactness of $\s C$ guarantees (using the extreme value theorem) the supremum is attained by some $\sigma_{x}^{U} \in \s C$:
    \begin{equation}
        \label{eq:least_divergent_sigma}
        \Delta_x(U^{\dagger} \sigma_x^{U} U) = \sup_{\sigma \in \s C} \Delta_x(U^{\dagger} \sigma U).
    \end{equation}
    Since $\infkeyl{\rho} = \keyl{\rho}{\sigma_x^{U}}$ and $\rho \neq \sigma_x^{U}$, we conclude, from \cref{cor:keyl_div}, that $\infkeyl{\rho} \neq 0$.
\end{proof}

\end{document}